\documentclass[10pt,twoside,twocolumn]{IEEEtran}

\IEEEoverridecommandlockouts
\usepackage{cite}
\usepackage{amsmath,amssymb,amsfonts,bbm,bm, mathrsfs, mathtools, dsfont, yhmath}
\usepackage{algorithmic}
\usepackage{algorithm}
\usepackage{graphicx}
\usepackage{subcaption}
\usepackage[font=footnotesize]{caption}
\usepackage{textcomp}
\usepackage{balance}
\usepackage[table, xcdraw]{xcolor}
\def\BibTeX{{\rm B\kern-.05em{\sc i\kern-.025em b}\kern-.08em
		T\kern-.1667em\lower.7ex\hbox{E}\kern-.125emX}}
\usepackage{array}

\usepackage{tikz}
\usepackage[utf8]{inputenc}
\usepackage{pgfplots}
\pgfplotsset{compat=newest}
\usepgfplotslibrary{groupplots}
\usepgfplotslibrary{dateplot}
\usepackage[caption=false,font=footnotesize]{subfig}
\usepackage{adjustbox}
\usepackage{amsmath}
\usepackage{amsthm}
\usepackage{cite}
\usepackage{microtype} 
\usepackage{balance}

\usepackage[printonlyused,withpage]{acronym}
\usepackage{times}
\usepackage{iipnotation}

\newacro{IOT}[IoT]{Internet-of-Things}

\newacro{IS}[IS]{intelligent surface}

\acrodef{ISAC}[ISAC]{integrated sensing and communication}

\acrodef{2D}[2D]{two-dimensional}
\acrodef{3D}[3D]{three-dimensional}

\acrodef{SNR}[SNR]{signal-to-noise ratio}

\acrodef{NN}[NN]{neural network}
\acrodef{DL}[DL]{downlink}
\acrodef{UL}[UL]{uplink}
\acrodef{OMP}[OMP]{orthogonal matching pursuit}
\acrodef{ISAC}[ISAC]{integrated sensing and communication}
\acrodef{MIMO}[MIMO]{multiple-input multiple-output}
\acrodef{OFDM}[OFDM]{orthogonal frequency-division multiplexing}

\acrodef{AoA}[AoA]{angle of arrival}
\acrodef{AoD}[AoD]{angle of departure}

\acrodef{RoI}[RoI]{region of interest}

\acrodef{CSI}[CSI]{channel state information}
\acrodef{FD}[FD]{full-duplex}
\acrodef{AWGN}[AWGN]{additive white Gaussian noise}

\acrodef{MU}[MU]{multi-user}

\acrodef{UE}[UE]{user equipment}
\acrodefplural{UE}[UEs]{user equipment}
\acrodef{CIS}[CIS]{continuous intelligent surface}
\acrodefplural{CIS}[CISs]{continuous intelligent surfaces}
\acrodef{DIS}[DIS]{discrete intelligent surface}
\acrodefplural{DIS}[DISs]{discrete intelligent surfaces}

\acrodef{Rx}{receiver}
\acrodef{Tx}{transmitter}

\acrodef{DFT}[DFT]{discrete Fourier transform}
\acrodef{IDFT}[IDFT]{inverse discrete Fourier transform}
\acrodef{E2E}[E2E]{end-to-end}
\acrodef{RG}[RG]{resource grid}
\acrodef{RB}[RB]{resource block}
\acrodef{LDPC}[LDPC]{low-density parity-check}
\acrodef{QAM}[QAM]{quadrature amplitude modulation}

\acrodef{CIR}[CIR]{channel impulse response}
\acrodef{DC}[DC]{direct-current}
\acrodef{BS}[BS]{base station}
\acrodef{ULA}[ULA]{uniform linear array}

\acrodef{BCE}[BCE]{binary cross entropy}
\acrodef{CCE}[CCE]{categorical cross entropy}

\acrodef{LS}[LS]{least-square}
\acrodef{SVD}[SVD]{singular value decomposition}

\acrodef{MLP}[MLP]{multi-layer perceptron}
\acrodef{CNN}[CNN]{convolutional neural network}

\tikzset{every picture/.append style={line join=round,line cap=round,}}
\tikzstyle{solid}=[dash pattern=]
\tikzstyle{dotted}=[dash pattern=on 0.0\pgflinewidth off 2.5\pgflinewidth]
\tikzstyle{dashed}=[dash pattern=on 4.0\pgflinewidth off 3.0\pgflinewidth]
\tikzstyle{dashdotted}=[dash pattern=on 0.0\pgflinewidth off 2.0\pgflinewidth on 3.0\pgflinewidth off 2.0\pgflinewidth]
\pgfplotsset{every tick label/.append style={font=\large}}
\pgfplotsset{every axis legend/.append style={font=\Large}}
\pgfplotsset{every axis label/.append style={font=\Large}}

\usepackage{tikz}
\usetikzlibrary{shapes.geometric, arrows}
\usetikzlibrary{ arrows.meta, calc, positioning, quotes,}
\usepackage{circuitikz}
\usetikzlibrary{arrows.meta,positioning}
\usepackage{adjustbox}
\usepackage{todonotes}
\usepackage{soul}

\definecolor{bl1}{HTML}{F7FCF0}
\definecolor{bl2}{HTML}{E0F3DB}
\definecolor{bl3}{HTML}{CCEBC5}
\definecolor{bl4}{HTML}{A8DDB5}
\definecolor{bl5}{HTML}{7BCCC4}

\IEEEoverridecommandlockouts

\tikzstyle{startstop} = [rectangle, rounded corners, minimum width=0.5cm, minimum height=1cm,text centered, draw=black, fill=red!30]
\tikzstyle{io} = [trapezium, trapezium left angle=70, trapezium right angle=110, minimum width=0.5cm, minimum height=1cm, text centered, draw=black, fill=blue!30]
\tikzstyle{process} = [rectangle, minimum width=0.5cm, minimum height=1cm, text centered, text width=3cm, draw=black, fill=orange!30]
\tikzstyle{decision} = [diamond, minimum width=0.5cm, minimum height=1cm, text centered, draw=black, fill=green!30]
\tikzstyle{arrow} = [very thick,->,>=latex]
\tikzstyle{invisible} =[rectangle, node distance=3cm, text=white]

\usetikzlibrary{external}
\newcommand{\nH}{{\mathrm{H}}}
\newcommand{\nT}{{\mathrm{T}}}
\newcommand{\SceneColSet}{{\Set \Omega}} 
\newcommand{\LengthInfoBit}{{L_{\rm info}}}

\newcommand{\op}[1]{\mathbb{#1}}

\newtheorem{proposition}{Proposition}
\theoremstyle{remark}
\newtheorem{remark}{Remark}
\newtheorem{lemma}{Lemma}
\newtheorem{definition}{Definition}

\begin{document}
	
	\title{
		Neural Integrated Sensing and Communication\\ for the MIMO-OFDM Downlink
	}
	
	\author{
		\IEEEauthorblockN{
			Ziyi~Wang, Frederik~Zumegen, and Christoph~Studer 
		}
		\thanks{}
				\thanks{The work of FZ and CS was supported in part by the Swiss National Science Foundation (SNSF) grant 200021\_207314 and by CHIST-ERA grant for the project CHASER (CHIST-ERA-22-WAI-01) through the SNSF grant 20CH21\_218704. The work of CS was also supported by the European Commission within the context of the project 6G-REFERENCE (6G Hardware Enablers for Cell Free Coherent Communications and Sensing), funded under EU Horizon Europe Grant Agreement 101139155.}
		\thanks{Ziyi~Wang was with the Department of Information Technology and Electrical Engineering, ETH Zurich, 8092 Zurich, Switzerland. He is now with the Institute of Electrical and Micro Engineering, EPFL, 1015 Lausanne, Switzerland. (email: ziyi.wang@epfl.ch)
			
			Frederik~Zumegen, and Christoph~Studer are with the Department of Information Technology and Electrical Engineering, ETH Zurich, 8092 Zurich, Switzerland (email: fzumegen@ethz.ch, studer@ethz.ch).

		}
	}
	
	\maketitle
	
	\begin{abstract}
		
		The ongoing convergence of spectrum and hardware requirements for wireless sensing and communication applications has fueled the \ac{ISAC} paradigm in next-generation networks. Neural-network-based \ac{ISAC} leverages data-driven learning techniques to add sensing capabilities to existing communication infrastructure.
		This paper presents a novel signal-processing framework for such neural \ac{ISAC} systems based on the \ac{MIMO} and \ac{OFDM} downlink. Our approach enables generalized sensing functionality without modifying the \ac{MIMO}-\ac{OFDM} communication link. Specifically, our neural \ac{ISAC} pipeline measures the backscattered communication signals to generate discrete map representations of spatial occupancy, formulated as multiclass or multilabel classification problems, which can then be utilized by specialized downstream tasks. 
		To improve sensing performance in closed or cluttered environments, our neural \ac{ISAC} pipeline relies on features specifically designed to mitigate strong reflective paths.
		Extensive simulations using ray-tracing models demonstrate that our neural \ac{ISAC} framework reliably reconstructs scene maps without altering the \ac{MIMO}-\ac{OFDM} communication pipeline or reducing data rates.
		
	\end{abstract}
	

	\section{Introduction}
	\acresetall
	\IEEEPARstart{S}{ensing} plays a key role in next-generation wireless networks.
	Potential applications include location-based services \cite{WanLiuSheConWin:J22, PatAshKypHerMosCor:M05, WanLiuSheConWin:J24}, object detection and tracking \cite{BarConWin:J17, ShaSauBucVar:J19}, smart environment \cite{CarFosBelCorBorTalCur:M13, ZanBuiCasVanZor:J14, PasBurFelZabDecVerAnd:J18, MorZamSka:J16}, and Internet-of-Things (IoT) \cite{LuoHoaWanNiyKimHan:J16, CheThoJarLohAleLepBhuBupFerHonLinRuoKorKuu:J17, WinMeyLiuDaiBarCon:J18, MioSicDepChl:J12, NagZhaNev:J17}. 
	In order to achieve high data rate and accurate sensing performance, both wireless sensing and communication systems build upon wider frequency bands and larger antenna apertures with multiple antennas~\cite{LiuHuaLiWanLiHanLiuDuTanLuSheColChe:J22}.  
	However, both the spectrum resources and large-aperture antenna arrays inevitably incur significant implementation and deployment costs \cite{CuiLiuJinMu:M21, LiuCuiMasXuHanEldBuz:J22}.
	At the same time, with the continuous evolution of sensing and communication systems, their properties in spectrum resource and hardware requirements are increasingly similar.
	Therefore, a promising approach to mitigating spectrum limitations and reducing hardware as well as deployment costs integrates sensing and communication into a unified system, which is the foundation of \emph{\acl{ISAC}} (\acs{ISAC}). 
	
	\acused{ISAC}

	\begin{figure}[!t]
		\centering
		\includegraphics[width=0.95\columnwidth]{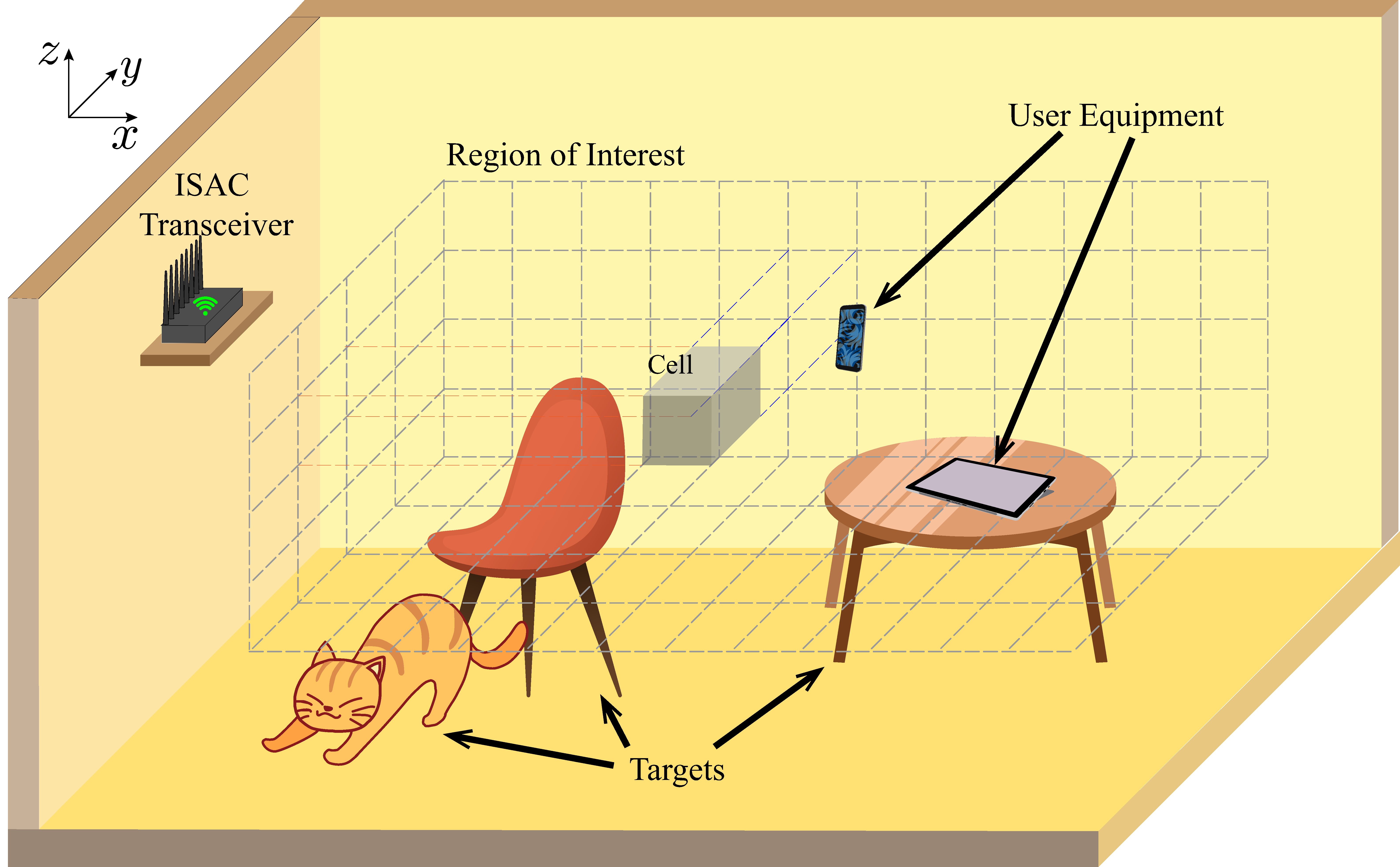}
		\caption{Example of an indoor scenario in which a neural ISAC transceiver simultaneously senses a given region of interest while maintaining one or several communication links.}
		\label{fig:isac_indoor_scenarios}
	\end{figure}

	\subsection{Contributions}
	This paper studies \ac{NN}-based \ac{ISAC} (neural \ac{ISAC}, for short).
	The fundamental questions related to neural \ac{ISAC} are as follows:
	(i) How to represent the spatial information of a scene?
	(ii) How to perform sensing without modifying the existing communication network?
	(iii) How to extract effective features for \ac{NN} training in neural \ac{ISAC}?
	This paper answers these questions with a novel paradigm for neural \ac{ISAC} built upon existing \ac{MIMO} and \ac{OFDM} communication networks.
	Furthermore, we advocate for the integration of sensing functionality in existing communication systems without any modification to the network.
	
	This paper proposes a novel signal processing framework that effectively senses the physical environment by building upon an existing \ac{MIMO}-\ac{OFDM} communication downlink.
	An indoor example of the paradigm for neural \ac{ISAC}  is shown in Fig.~\ref{fig:isac_indoor_scenarios}. 
	The physical objects in this scene (table, chair, cat, etc.) become targets that can be sensed seamlessly with our proposed neural \ac{ISAC} system.
	The key contributions of the paper are summarized as follows:
	\begin{itemize}
		\item We introduce the concept of discrete map representations to describe scenes. We then formulate a task-agnostic sensing problem, namely, scene reconstruction under these representations,\footnote{
			The sensing task proposed in this paper is similar to the imaging task in \cite{MehSab:J22}, but we focus on \ac{DL} scenarios. Additionally, in some papers,  e.g., \cite{TorGueZhaGuiYanEldDar:J24}, the term imaging implies that the sensed channel is already known. To avoid such ambiguities, this paper uses the terms sensing or scene reconstruction instead of imaging.
		} into multilabel or multiclass classification.
		\item We establish a signal processing framework for scene reconstructions by reusing the \ac{MIMO}-\ac{OFDM} downlink. Specifically, we introduce Tikhonov regularization in channel estimation, and perform feature extraction fusing environmental prior knowledge.
		\item We conduct case studies to identify effective features and fusion techniques and quantify the sensing performance of the systems.
	\end{itemize}

	\subsection{Prior Art}
	In order to unify sensing and communication under a joint signal processing framework, some existing papers have explored orthogonal resource allocations of waveforms and waveform designs with higher integration. 
	Orthogonal resource allocation in \ac{ISAC} systems divides resources in time, frequency, or space, and assigns them to sensing and communication tasks separately \cite{NaoBazBomCha:J24}; this ensures low interference between sensing and communication signals, but also limits the potential benefits of integration. 
	Waveform designs can be categorized into communication-centric, sensing-centric, and joint approaches.  
	Communication- and sensing-centric approaches use waveforms originally designed for communication and sensing to perform sensing and communication tasks, respectively \cite{StuWie:11}; while joint waveform designs jointly optimize the sensing and communication performance by devising the \ac{MIMO} precoders \cite{LiuZhoMasLiLuoPet:J18, LiuLiuLiMasEld:J22, LuLiuDonXioXuLiuJin:J24}. Unified waveforms can achieve better fundamental trade-offs of ISAC systems. However, the specialized waveform design modifies the original communication links inevitably, which may lead to additional communication latency and hardware overhead.
	Therefore, it is more promising to apply communication-centric waveforms and integrate sensing into communication networks without modifying existing communication links or involving orthogonal resource allocation for sensing and communication. 
	In this paper, we focus on \ac{DL} ISAC systems with standard \ac{MIMO}-\ac{OFDM} communication networks. 
	
	Learning-based methods have shown to hold great promise in a range of sensing tasks in \ac{ISAC} systems \cite{MatHagKesLemWym:C23, MatHagKesMagWym:25, NaoBazBomCha:J24, JiaMaWeiFenZhaPen:J24, CheFenZhaGaoYuaYanZha:J24}. 
	These methods can be categorized into mainly two types: model-based learning and data-driven learning. 
	Model-based learning for \ac{ISAC} systems originates from classical algorithms of statistical inference and compressive sensing, such as message passing and \ac{OMP} algorithms \cite{MatHagKesLemWym:C23, MatHagKesMagWym:25, JiaMaWeiFenZhaPen:J24}. 
	For this reason, the methods based on model-based learning are generally more interpretable compared to those based on data-driven learning. 
	However, model-based learning also imposes implicit constraints on sensing tasks. 
	Such constraints limit most existing approaches of \ac{ISAC} using model-based learning to the estimation of several parameters \cite{MatHagKesLemWym:C23, MatHagKesMagWym:25, JiaMaWeiFenZhaPen:J24}, e.g., \ac{AoA}, \ac{AoD}, and range, which can lead to information loss. 
	In contrast, while some existing results on data-driven learning for \ac{ISAC} are also limited to parameter estimation of \ac{AoA}, \ac{AoD}, and range \cite{MatHagKesLemWym:C23, NaoBazBomCha:J24, MatHagKesMagWym:25}, the methods based on data-driven learning are not theoretically bound by these constraints and, thus, allow for more general sensing tasks. In particular, neural \ac{ISAC}, empowered by \acp{NN}, stands out as a strong candidate due to the model-fitting capability of  \acp{NN} \cite{Hor:J91, HanRol:J19, CheFenZhaGaoYuaYanZha:J24, NaoBazBomCha:J24}. In this paper, the proposed signal processing framework is based on data-driven learning and directly reconstructs the spatial information of scenes under discrete map representations.

	\subsection{Notation}
	Throughout this paper, scalars, vectors and matrices are displayed in boldface uppercase, boldface lowercase and italic, respectively. For example, a scalar, a vector, and a matrix are denoted by $x$, $\V x$, and $\M X$, respectively. Sets and tensors are denoted by calligraphic font and boldface calligraphic font, respectively. For example, a set and a tensor can be denoted by $\Set X$ and $\Tsr X$.
	Let $\M X^*$, $\M X^\nT$, $\M X^\nH$, and $\M X^\dagger$ denote the conjugate, transpose, conjugate transpose and pseudo-inverse of matrix $\M X$, respectively. 
	The $(i,j)$th element, $i$th row, $j$th column, and rank of a matrix $\M X$ are denoted by $\left[\M X\right]_{i,j}$, $\left[\M X\right]_{i,:}$, $\left[\M X\right]_{:, j}$, and $\mathrm{rank}\left(\M X\right)$ respectively. 
	The $i$th horizontal slice, $j$th lateral slice, and $k$th frontal slice of a \ac{3D} tensor $\Tsr X$ are denoted by $\left[\Tsr X\right]_{i,:,:}$, $\left[\Tsr X\right]_{:,j,:}$, and $\left[\Tsr X\right]_{:,:,k}$, respectively.
	The $m $-by-$n$ matrix of zeros (resp. ones) is denoted by $\V 0_{m\times n}$ (resp. $\V 1_{m\times n }$); when $n=1$, the $m$th dimensional vector of zeros (resp. ones) is simply denoted by $\V 0_m$ (resp. $\V 1_m$).
	The $m$-by-$m$ identity matrix is denoted by $\M I_m$: the subscript is removed when the dimension of the matrix is clear for the context.
	The $m$-dimensional unitary \ac{DFT} matrix is denoted by $\M F_m \in \mathbb{C} ^{m\times m}$.
	The operator $\left\Vert \cdot \right\Vert_{\rm F}$ takes the Frobenius norm of its argument. 
	The operator $\circ$ denotes the composition of two functions. For example, the composition of the functions $f(\cdot)$ and $g(\cdot)$ is denoted by $f \circ g$, i.e., $(f\circ g)(\cdot) =f(g(\cdot))$.
	The identity operator is denoted by $\op I_{\rm d}$.
	The operator $\times_n$ denotes the mode-$n$ tensor product \cite{Zha:B17}. The concatenation of the tensors $\Tsr A$ and $\Tsr B$ along the $n$th dimension is denoted by $\begin{bmatrix}
		\Tsr A & \Tsr B
	\end{bmatrix}_n$.
	The operators $\Re\{\cdot\}$ and $\Im\{\cdot\}$ take the real and imaginary part of their arguments; the imaginary unit is denoted by $\jmath$.
	The indicator function of a set $\Set X$ is denoted by $\indicator{\Set X}$, i.e., $\indicator{\Set X}(x) = 1$ if $x \in \Set X$ and $\indicator{\Set X}(x) = 0$ if $x \notin \Set X$.
	The operators $\vee$ and $\wedge$ denote the logical OR and AND operation, respectively, i.e., $x \vee y = 1$ if $x>0$ or $y>0$, and $x \vee y = 0$ if $x = y = 0$; $x \wedge y = 0$ if $x =0$ or $y=0$, and $x\wedge y = 1$ if $x, y>0$.
	The notation of important quantities is summarized in Table~\ref{tab:notation}.

	\begin{table*}
		
		\caption{Notation summary of important quantities.}
		\renewcommand{\arraystretch}{1.5}
		\begin{center}
			\resizebox{1\textwidth}{!}{\begin{tabular}{ p{1cm} p{8.5cm} | p{1cm} p{8.5cm} }
					\hline
					\bf Notation & \bf Definition & \bf Notation & \bf Definition \\
					\hline
					\rowcolor{blue!10!white}
					$\SceneColSet$ &Collection representing a scene composed of targets within the \ac{RoI}& 
					$\SceneColSet_j$  & Set corresponding to the $j$-th target in the scene \\
					$N_{\rm r}$& Number of \ac{Rx} antennas at the \ac{ISAC} transceiver&
					$N_{\rm t} $& Number of \ac{Tx} antennas at the \ac{ISAC} transceiver\\
					\rowcolor{blue!10!white}
					$W$& Number of \ac{OFDM} subcarriers& 
					$K$& Number of \ac{MIMO} datastreams\\
					$L$& Number of \ac{OFDM} symbols&
					$\LengthInfoBit$& Length of information bits transmitted in the communication link  
					
					\\
					\rowcolor{blue!10!white} 
					$\M H_w^{\rm c}$ & \ac{DL} communication channel matrix at the $w$th subcarrier& 
					$\M H_w^{\rm s}$ & \ac{DL} sensing channel matrix at the $w$th subcarrier \\
					$\M P_w$ & \ac{MIMO} precoding matrix designed for communication at the $w$th subcarrier& 
					$\M S_w$ & Transmitted \ac{OFDM} data matrix containing $L$ \ac{OFDM} symbols for $K$ single-antenna \acp{UE}\\
					
					\rowcolor{blue!10!white}
					$\V b^{(q)}$ & Binary information bits associated to the $q$th data sample&
					$\hat{\V b}^{(q)}$ & Detected binary information bits associated to the $q$th data sample\\
					$\V m^{(q)}$& The discrete map representation of the scene associated to the $q$th data sample&
					$\hat{\V m}^{(q)}$& Estimated discrete map representation of the scene associated to the $q$th data sample\\
					\rowcolor{blue!10!white}
					$\Tsr H^{{\rm c}(q)}$& Tensor of communication \ac{CSI} associated to the $q$th data sample & 
					$\Tsr H^{{\rm s}(q)}$& Tensor of sensing \ac{CSI} associated to the $q$th data sample  \\
					$\Tsr P^{(q)}$ & \ac{MIMO} precoding tensor associated to the $q$th data sample & 
					$\Tsr S^{(q)}$& Tensor of transmitted datastreams associated to the $q$th data sample \\
					
					\rowcolor{blue!10!white}
					$\Tsr F_{\rm bd}^{\rm s}$ & Beamspace and delay-domain feature tensor extracted from the estimated sensing \ac{CSI}& 
					$\Tsr F_{\rm d}^{\rm s}$& Direct feature tensor extracted from the estimated sensing \ac{CSI}\\
					$\Tsr F_{\rm e}^{\rm s}$& Feature tensor extracted from the estimated sensing \ac{CSI}& 
					$\Tsr F_{\rm e, ref}^{\rm s} $& Feature tensor extracted from the reference sensing \ac{CSI}\\
					
					\rowcolor{blue!10!white}
					$\lambda _{\rm reg}$& Tikhonov regularization factor for sensing channel estimation& 
					$g_{\rm fus}^{\rm SUB}$& Subtraction fusion mapping of the feature tensors respectively extracted from the estimated and reference \ac{CSI} \\
					$g_{\rm fus}^{\rm STA}$& Stacking fusion mapping of the feature tensors respectively extracted from the estimated and reference \ac{CSI}  & 
					$g_{\rm fus}^{\rm NOR} $& Fusion mapping that drops the feature tensor extracted from the reference \ac{CSI}\\
					\hline
			\end{tabular}}
			\vspace{-2mm}
		\end{center}
		\label{tab:notation}
	\end{table*}

	\subsection{Paper Outline}
	
	The rest of this paper is organized as follows.
	Section~\ref{sec:sysmod} presents the signal models for the considered \ac{DL} \ac{ISAC} systems.
	Section~\ref{sec:sc_nisac}  formulates the neural \ac{ISAC} problem.
	Section~\ref{sec:sp_flow} elaborates on the proposed signal processing framework for neural \ac{ISAC}.
	Section~\ref{sec:casestudy} conducts case studies to illustrate the concepts of scene reconstruction under discrete map representations, evaluates the performance of the neural \ac{ISAC} systems under the proposed framework, and discusses the primary limitation in practical deployments.
	Section~\ref{sec:conclusion} concludes.

	\section{Communication and Sensing Models}
	\label{sec:sysmod}
	
	This section details the signal models for the communication and sensing channels. Fig.~\ref{fig:neurisac_blockdiagram} shows the system diagram of the neural \ac{ISAC} system. The communication channel refers to the channel between the \ac{Tx} at the \ac{ISAC} transceiver and the \acp{UE}. The sensing channel refers to the channel between the \ac{Tx} and \ac{Rx} at the \ac{ISAC} transceiver, which contains the information for scene reconstruction. The \ac{Tx} and \ac{Rx} antennas at the \ac{ISAC} transceiver can be either separate or configured as a \ac{FD} entity.
	
	\subsection{Signal Model for Communication}\label{ssec:sglmod:comm}
	We consider a \ac{MIMO}-\ac{OFDM} \ac{Tx} with $N_{\rm t}$ antennas communicating with $K$ single-antenna \acp{UE}. For the $w$th subcarrier, the received $L$ \ac{OFDM} symbols $\M Y^{\rm c}_{w}\in \mathbb{C} ^{K \times L}$ at \acp{UE} in the frequency domain are given by
	\begin{IEEEeqnarray}{C"l}
		\M Y_w^{\rm c} = \M H_w^{\rm c} \M P_w \M S_w + \M Z_w^{\rm c}, & w \in \Set W \triangleq\{1, \dotsc, W\}\,,
	\end{IEEEeqnarray}
	where the matrix $\M H_w^{\rm c} \in \mathbb{C} ^{K \times N_{\rm t}}$ is the communication channel for the $w$th subcarrier between the \ac{Tx} and \acp{UE}; 
	the matrix $\M S_w \in \mathbb{C} ^{K \times L}$ contains the transmitted datastreams to the $K$ \acp{UE};
	the matrix $\M P_w \in \mathbb{C} ^{N_{\rm t} \times K}$ is the precoding matrix associated to the communication channel matrix $\M H_w^{\rm c}$, and its $k$th column $\left[\M P_w\right]_{:, k}$ is the precoding vector for the $k$th \ac{UE}; and
	the matrix $\M Z_w^{\rm c} \in \mathbb{C} ^{K \times L}$ presents complex \ac{AWGN}, with elements that are independent and identically distributed, following a zero-mean circularly-symmetric complex Gaussian distribution with a variance of~$N_0$ per entry. 
	
	\subsection{Signal Model for Sensing}
	\label{ssec:sglmod:sens}
	We consider a \ac{MIMO}-\ac{OFDM} transceiver with $N_{\rm t}$ \ac{Tx} antennas and $N_{\rm r}$ \ac{Rx} antennas. The $N_{\rm t}$ \ac{Tx} antennas communicate with $K$ single-antenna \acp{UE}.  The $N_{\rm r}$ \ac{Rx} antennas receive the signals excited by the transmitted signals for \ac{DL}  \ac{MU}-\ac{MIMO} communication, and the received $L$ \ac{OFDM} symbols $\M Y_w^{\rm s} \in \mathbb{C} ^{N_{\rm r} \times L}$ at the $w$th subcarrier are given by
	\begin{IEEEeqnarray}{rCl}
		\M Y_w^{\rm s} = \M H_w^{\rm s} \M P_w \M S_w + \M Z_w^{\rm s}, \qquad w \in \Set W\,.
		\label{eq:sys_model:sens:w}
	\end{IEEEeqnarray}
	where the matrix $\M H_w^{\rm s} \in \mathbb{C} ^{N_{\rm r}\times N_{\rm t}}$ is the sensing channel matrix at the $w$th subcarrier, and the matrix $\M Z_w^{\rm s} \in \mathbb{C} ^{N_{\rm r} \times L}$  presents complex \ac{AWGN}, with elements that are independent and identically distributed, following a zero-mean circularly-symmetric complex Gaussian distribution with a variance of $N_0$ per entry.
	
	For convenience in the following sections, the received sensing signal model in \eqref{eq:sys_model:sens:w} is expressed in tensor form, i.e.,
	\begin{IEEEeqnarray}{rCl}
		\Tsr Y^{\rm s} = \Tsr H^{\rm s} \times_3 \Tsr P \times_3 \Tsr S + \Tsr Z^{\rm s}\,,
		\label{eq:signal_model_sens:tensor}
	\end{IEEEeqnarray}
	where the received sensing signal tensor $\Tsr Y^{\rm s} \in \mathbb{C} ^{N_{\rm r} \times L \times W}$, the sensing channel tensor $\Tsr H^{\rm s} \in \mathbb{C} ^{N_{\rm r}\times N_{\rm t} \times W}$, the \ac{MIMO} precoding tensor $\Tsr P \in \mathbb{C} ^{N_{\rm t} \times K \times W}$, the transmitted datastream tensor $\Tsr S \in \mathbb{C} ^{K \times L \times W}$, and the noise tensor $\Tsr Z^{\rm s} \in \mathbb{C} ^{N_{\rm r}\times L \times W}$ are concatenated along the third dimension from the corresponding matrices indexed by $w \in \Set W$. 
	For example, the tensor $\Tsr Y^{\rm s}$ is specifically given by
	\begin{IEEEeqnarray}{rCl}
		\begin{bmatrix}
			\M Y_1^{\rm s}&\M Y_2^{\rm s}& \cdots & \M Y_W^{\rm s}
		\end{bmatrix}_3\,,
	\end{IEEEeqnarray}
	and the tensor product $\Tsr H^{\rm s}\times_3 \Tsr P \times_3 \Tsr S$ is equivalent to
	\begin{IEEEeqnarray}{rCl}
		\begin{bmatrix}
			\M H_1^{\rm s} \M P_1 \M S_1 & \M H_2^{\rm s} \M P_2 \M S_2 & \cdots & \M H_W^{\rm s} \M P_W \M S_W
		\end{bmatrix}_3\,.
	\end{IEEEeqnarray}

	\begin{remark} \label{rmk:fullduplex}
		The signal model of \ac{DL} sensing with \ac{MU}-\ac{MIMO} in this section is applicable to a \ac{FD} \ac{ISAC} transceiver. In this case, the number $N_{\rm t}$ of Tx antennas at the \ac{ISAC} transceiver is equal to the number $N_{\rm  r}$ of \ac{Rx} antennas. 
	\end{remark}
	
	\begin{figure*}[!t]
		\centering
		\includegraphics[width=0.99\textwidth]{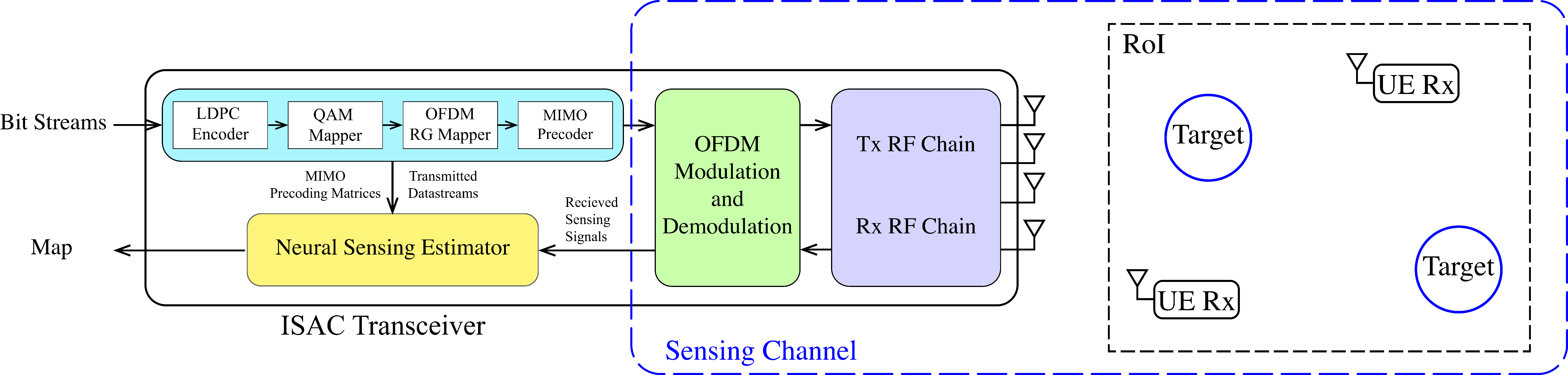}
		\caption{System diagram of the neural \ac{ISAC} system. The \ac{ISAC} transceiver performs \ac{DL} communication with multiple \acp{UE}, while simultaneously extracting scene information from the received sensing signals via the sensing channel, the \ac{MIMO} precoding matrices, and the transmitted datastreams.}
		\label{fig:neurisac_blockdiagram}
	\end{figure*}
	
	\section{Standard-Compatible Neural ISAC}\label{sec:sc_nisac}
	This section first proposes discrete map representations for spatial sensing. Then, the general learning problem of neural \ac{ISAC} systems is formulated mathematically and further decomposed into two problems, namely, standard-compatible \ac{MU}-\ac{MIMO}-\ac{OFDM} \ac{DL} communication and neural sensing under discrete map representations. Furthermore, the neural sensing problems under different discrete map representations are formulated in detail as multiclass classification and multilabel classification, respectively.
	\subsection{Discrete Map Representations}\label{ssec:dmr}
	
	\begin{figure*}[!t]
		\centering
		\begin{subfigure}{0.6\columnwidth}
			\centering
			\includegraphics[width = \textwidth]{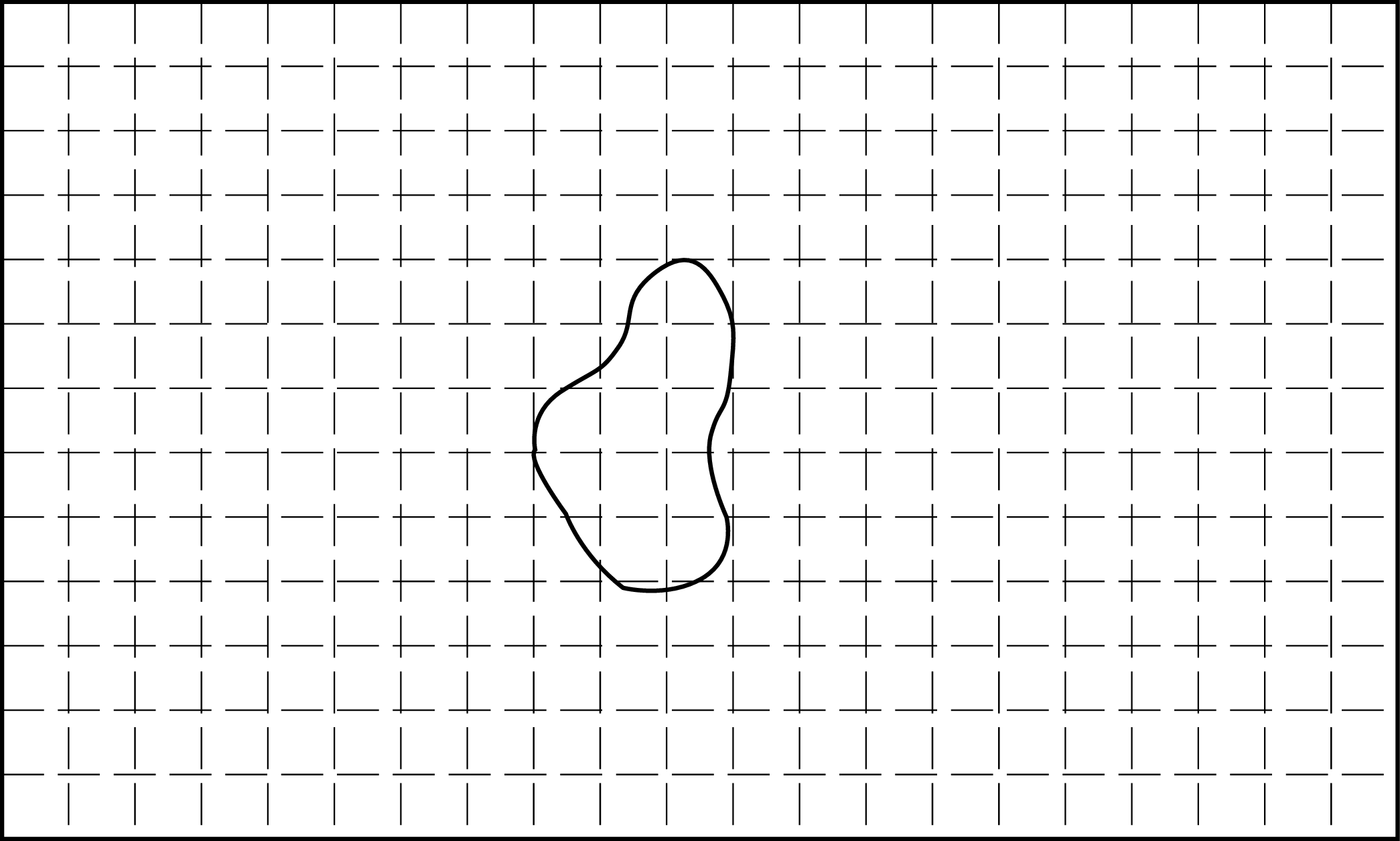}
			\caption{}
			\label{fig:ground_truth}
		\end{subfigure}
		\hfill
		\begin{subfigure}{0.6\columnwidth}
			\centering
			\includegraphics[width = \textwidth]{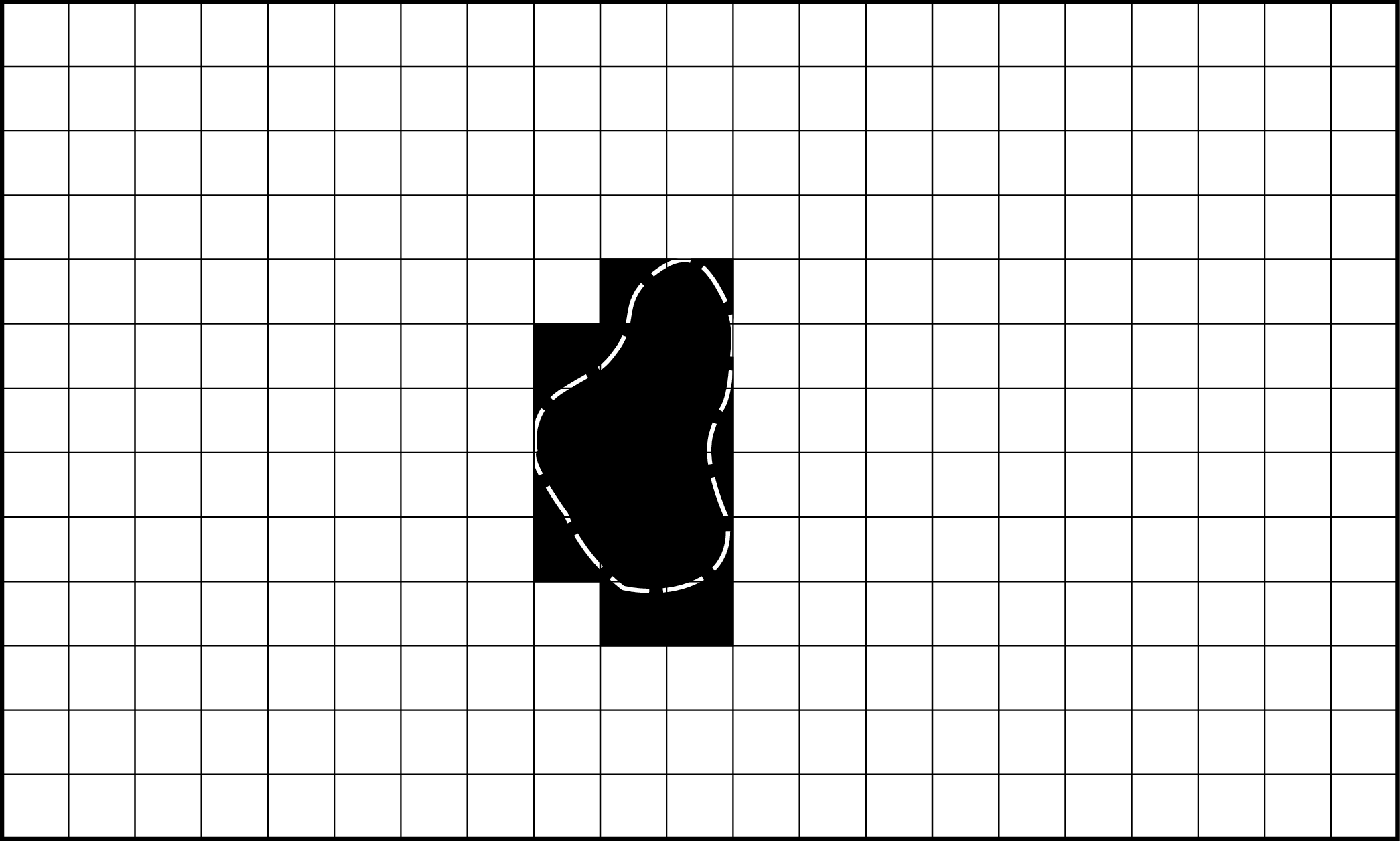}
			\caption{}
			\label{fig:hard_map}
		\end{subfigure}
		\hfill
		\begin{subfigure}{0.6\columnwidth}
			\centering
			\includegraphics[width = \textwidth]{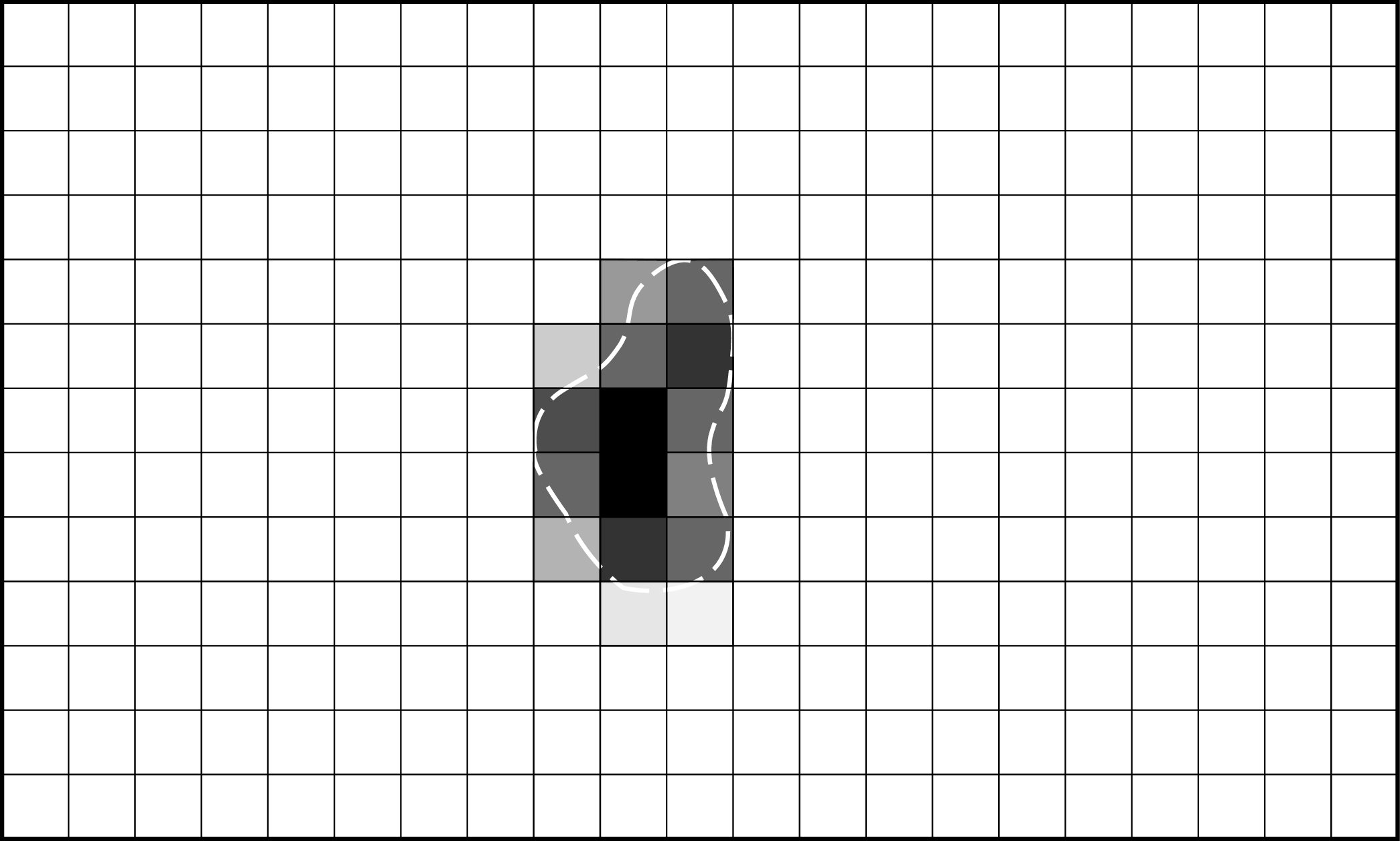}
			\caption{}
			\label{fig:soft_map}
		\end{subfigure}
		
		\caption{Illustration for map representations of scenes under a specific map grip. The gray levels of the cells represent the map values $m_i \in [0, 1]$. The white color corresponds to $0$, and the black color corresponds to $1$. (a) The considered scene with a map grid. (b) The hard map representation of the scene. (c) The soft map representation of the scene.}
		\label{fig:map_representation}
		
	\end{figure*}

	To describe the scene information in the \ac{RoI}, the map grid and its cells are defined as in Definition~\ref{def:map_grid_cell}.
	
	\begin{definition}[Map grid and cells] \label{def:map_grid_cell}
		The map grid $\Set G$ is defined as a collection of $I$ cells containing the \ac{RoI} $\Set R$, given by
		\begin{IEEEeqnarray}{rCl}
			\Set G \triangleq \left\{
			\Set G_i: \Set G_i \neq \varnothing, i \in \Set I = \{1, \dotsc, I\}
			\right\}\,,
		\end{IEEEeqnarray}
		where the cells $\Set G_i$, $i\in \Set I$ satisfy that 
		\begin{IEEEeqnarray}{r"C"l}
			\bigcup_{i \in \Set I} \Set G_i =\Set R&\text{and}&\bigcup_{\substack{i,i'\in \Set I \\ i \neq i'}} \left(\Set G_i \cap \Set G_{i'}\right) = \varnothing\,.
		\end{IEEEeqnarray}
	\end{definition}
	
	Based on a map grip $\Set G$, the hard map representation can be introduced as in Definition~\ref{def:hard_map}.
	
	\begin{definition}[Hard map representation]\label{def:hard_map}
		The hard map representation $\V m$ is defined as
		\begin{IEEEeqnarray}{C"l}
			m_i \triangleq \bigvee_{j \in \Set J} \bigvee_{\V p \in \SceneColSet_j}\indicator{\Set G_i} (\V p)\,,& i \in \Set I\,,
		\end{IEEEeqnarray}
		where the collection $\SceneColSet$ represents the targets and their shapes in the \ac{RoI} $\Set R$ of the scene,  given by
		\begin{IEEEeqnarray}{rCl}
			\SceneColSet = \left\{\SceneColSet_j\subseteq \Set R: \SceneColSet_j \neq \varnothing, j\in\Set J=\{1, \dotsc, J\}\right \} \label{eq:def_hard_map:scene_coll}
		\end{IEEEeqnarray}
		with $\SceneColSet_j \cap \SceneColSet_{j'} = \varnothing$ for any distinct $j,\, j' \in \Set J$.\footnote{This implies that the target shapes have no physical overlap.}
	\end{definition}
	
	\begin{remark}
		An example of a \ac{2D} hard map representation of a scene is depicted in Fig.~\ref{fig:map_representation}, where Fig.~\ref{fig:ground_truth} shows the considered scene with a map grid, and Fig.~\ref{fig:hard_map} shows its corresponding hard map representation. An intuitive interpretation of Definition~\ref{def:hard_map} is to assign each cell in the map grid the value $1$ if it overlaps with the scene, and $0$ otherwise.
		For the single-target case, if the shape information of the target is ignored, then the \ac{2D} hard map representation coincides with the probability map representation in equation (11) of \cite{GonLeiLanHuaStu:J22} with a sufficiently small $\varepsilon$. 
		In this paper, we call the hard map representation of the case with only one occupied cell a ``probability map representation.'' 
	\end{remark}
	
	Furthermore, the definition of soft map representation is presented in Definition~\ref{def:soft_map}.
	\begin{definition}[Soft map representation]\label{def:soft_map}
		The soft map representation $\tilde{\V m}$ is defined as
		\begin{IEEEeqnarray}{C"l}
			\widetilde{m}_i \triangleq \mu_{\SceneColSet}(\Set G_i), & i \in \Set I\,,
		\end{IEEEeqnarray}
		where the measure $\mu_{\SceneColSet}$ of set $\Set S$ is defined as
		\begin{IEEEeqnarray}{rCl}
			\mu_{\SceneColSet}(\Set S) \triangleq \left\{\begin{array}{ll}
				0																										 &\text{if } \mu(\Set S) = 0\\
				\frac{1}{\mu(\Set S)}\sum_{j \in \Set J} \mu(\Set S \cap \SceneColSet_j) & \text{otherwise},
			\end{array}\right.
		\end{IEEEeqnarray}
		in which the measure $\mu$ yields the volume  or area  of the set, i.e.,
		\begin{IEEEeqnarray}{rCl}
			\mu(\Set S) \triangleq 
			\left\{\begin{array}{ll}
				\lambda_3 (\Set S)																								 &\text{if } \lambda_3(\Set S) > 0\\
				H^2(\Set S) & \text{if } \lambda _3(\Set S)=0,
			\end{array}\right.
		\end{IEEEeqnarray}
		with the \ac{3D} Lebesgue measure $\lambda_3$ and the \ac{2D} Hausdorff measure $H^2$.
	\end{definition}
	
	\begin{remark}
		An example of the \ac{2D} soft map representation of the scene in Fig.~\ref{fig:ground_truth} is also included in Fig.~\ref{fig:map_representation}, where Fig.~\ref{fig:soft_map} shows the corresponding soft map representation. An intuitive interpretation of Definition~\ref{def:soft_map} is to assign each cell in the map grid the proportion of the cell that is occupied by the scene.
		Soft map representations show the confidence level of each cell to be occupied by targets on the map grid. 
		Under the same map resolution, soft map representations provide more information on the boundary of the targets compared to hard map representations, which is critical for spatial sensing.
	\end{remark}
	
	\subsection{General Problem Formulation of Neural \ac{ISAC}}
	To integrate sensing functionality into standard \ac{MIMO}-\ac{OFDM} networks, the neural \ac{ISAC} problem is transformed into a \ac{NN} training problem for a neural sensing estimator at the \ac{ISAC} transceiver.  
	To illustrate this transformation, we consider the supervised learning problem of the considered \ac{DL} neural \ac{ISAC} system for the \ac{RoI} $\Set R$ over the dataset 
	\begin{IEEEeqnarray}{rCl}
		\Set D = \left\{\left(\V b^{(q)}, \SceneColSet^{(q)}\right) : \V b^{(q)} \in \left\{0,1\right\}^{\LengthInfoBit}, \SceneColSet^{(q)} \subset 2^{\Set R},  q \in \Set Q\right\}\!,  \nonumber\\ 
		\label{eq:dataset:abstract}
	\end{IEEEeqnarray}
	and the learning problem can be generally formulated as follows:
	\begin{IEEEeqnarray}{r'l'l}
		\mathscr{P}_0: &\mathop{\text{minimize}}\limits_{\V \theta  \in \Set \Theta}&
		\eta \E{\frac{1}{\left|\Set Q\right|}\sum_{q \in \Set Q}L^{\rm comm}(\V b^{(q)}, \hat{\V b}^{(q)})} \\
		&&
		+ (1- \eta) \E{\frac{1}{\left|\Set Q\right|}\sum_{q \in \Set Q}L^{\rm sens}(\Set  \Omega^{(q)}, \hat{\V m}^{(q)})} \nonumber\\
		&\text{subject to}& \begin{pmatrix}
			\hat{\V b}^{(q)}\\ \hat{\V  m}^{(q)}
		\end{pmatrix}= f_{\V \theta }\begin{pmatrix}
			\V b^{(q)}\\ \Set  \Omega^{(q)}
		\end{pmatrix}\,,\quad \forall  q \in \Set Q.
	\end{IEEEeqnarray}
	\noindent Here, the vector $\V \theta $ is the learnable parameter vector of the entire  neural \ac{ISAC} system $f_{\V \theta }\left(\cdot, \cdot\right)$ and its parameter space is denoted by $\Set \Theta$,\footnote{If there is no learnable parameter vector in the  \ac{ISAC} system, the learnable parameter vector is void.}
	the expectation is taken over the noise in the communication and sensing channels,
	$\eta$ is the trade-off factor of sensing and communication;
	vectors $\V b^{(q)}, \hat{\V b}^{(q)}\in \{0, 1\}^{\LengthInfoBit}$ are the transmitted information bits and their estimates associated to the $q$th data sample;
	collection $\SceneColSet^{(q)}$ is the set of the spaces occupied by targets in the RoI of the scene;
	the vector $\hat{\V m}^{(q)} \in[0,1]^{\left|\Set I\right|}$ is the discrete map representation of the targets distributed over the considered map grid $\Set G = \{\Set G_i: i\in \Set I\}$;
	the functions $L^{\rm comm}(\cdot)$ and $L^{\rm sens}(\cdot)$ are loss functions for communication and sensing tasks, respectively.
	
	\begin{remark}\label{rmk:datapoint:input_eq_label}
		For problem $\mathscr{P}_0$, a data point $\left(\V b^{(q)}, \SceneColSet^{(q)}\right)$, $q \in \Set Q$, serves as both the input and the label for the output of the neural \ac{ISAC} system.
		During training, the vector $\V b^{(q)}$ is a length-$L_{\rm info}$ sequence of information bits, randomly generated from a binary source, and remains constant over time. We assume that, during the period from when $\V b^{(q)}$ is transmitted by the \ac{Tx} of the \ac{ISAC} transceiver to when the corresponding \ac{OFDM} signal is received by the \ac{Rx}, the scene $\Set \Omega^{(q)}$ undergoes no spatial displacement.
	\end{remark}
	
	\begin{remark}
		Under the proposed framework, the sensing task in ISAC is to reconstruct the scene in the \ac{RoI} with a specific resolution. The parameters estimated in conventional sensing tasks, e.g., \ac{AoA}, \ac{AoD}, and range, can be acquired from the estimated scene in discrete map representations.
	\end{remark}
	
	The \ac{DL}  \ac{ISAC} system $f_{\V \theta }$ can be decomposed into a combination of subsystems as shown in Fig.~\ref{fig:neurisac_blockdiagram}. Specifically, the decomposition can be expressed as
	\begin{IEEEeqnarray}{rCl}
		f_{\V \theta }\begin{pmatrix}
			\V b\\ \SceneColSet
		\end{pmatrix}= 
		\begin{pmatrix}
			g_{\V \theta_{\rm c}^{\rm r}} \circ h_{\rm c}\left(\SceneColSet, \cdot\right) \circ \tilde{g}_{\V \theta^{\rm t}} 
			\\
			g_{\V \theta _{\rm s}^{\rm r}}\left(\Tsr P_{\V \theta^{\rm t}}, \cdot, \cdot\right) \circ 
			\begin{pmatrix}
				\op I_{\rm d}
				\\
				h_{\rm s}(\SceneColSet, \cdot) \circ \tilde{g}_{\V \theta ^{\rm t}}
			\end{pmatrix}
		\end{pmatrix}  \circ g_{\rm c}^{\rm t}(\V b)\nonumber\\
	\end{IEEEeqnarray}
	where 
	the learnable parameter vector $\V \theta \in\Set \Theta $ of the \ac{DL}  \ac{ISAC} system is partitioned as
	\begin{IEEEeqnarray}{rCl}
		\V \theta  = \begin{bmatrix}
			\left(\V \theta _{\rm c}^{\rm r}\right)^\nT & 
			\left({\V \theta _{\rm s}^{\rm r}}\right)^\nT & \left(\V \theta ^{\rm t}\right)^\nT
		\end{bmatrix}^\nT\,;
	\end{IEEEeqnarray}
	the subsystem 
	$g_{\V \theta _{\rm c}^{\rm r}}: \mathbb{C} ^{K \times L \times W} \rightarrow \{0, 1\}^{\LengthInfoBit}$ 
	is the multiple \acp{UE} with learnable parameter vector $\V \theta _{\rm c}^{\rm r} \in \Set \Theta _{\rm c}^{\rm r}$; 
	the subsystem 
	$g_{\V \theta _{\rm s}^{\rm r}}: 
	\mathbb{C}^{N_{\rm t} \times K \times W} \times
	\mathbb{C} ^{K \times L\times W}\times
	\mathbb{C} ^{N_{\rm r} \times L \times W}
	\rightarrow{[0,1]^{\left|\Set I\right|}}$ 
	is the neural sensing estimator with learnable parameter vector $\V \theta _{\rm s}^{\rm r} \in \Set \Theta_{\rm s}^{\rm r}$ at the \ac{ISAC} transceiver, whose first input is a linear \ac{MIMO} precoding tensor $\Tsr P_{\V \theta^{\rm t}}$ associated to learnable parameter vector $\V \theta^{\rm t}$; 
	the subsystems
	$h_{\rm c}(\SceneColSet,\cdot): \mathbb{C} ^{N_{\rm t} \times L \times W} \rightarrow \mathbb{C} ^{K \times L \times W}$ and $h_{\rm s}(\SceneColSet, \cdot): \mathbb{C} ^{N_{\rm t} \times L \times W} \rightarrow \mathbb{C} ^{N_{\rm r} \times L \times W}$ are the communication and sensing channels, respectively, within the \ac{RoI} of the scene determined by the collection $\SceneColSet$, in which $\SceneColSet$ has a form similar to \eqref{eq:def_hard_map:scene_coll};
	the subsystem  $\tilde{g}_{\V \theta^{\rm t}}: \mathbb{C} ^{K \times L \times W} \rightarrow \mathbb{C} ^{N_{\rm t} \times L \times W}$ is the linear \ac{MIMO} precoder; and
	the subsystem $g^{\rm t}_{\rm c}: \{0,1\}^{\LengthInfoBit} \rightarrow \mathbb{C} ^{K \times L \times W}$ maps the binary source to the \ac{OFDM} \ac{RG}, composed of the \ac{LDPC} encoder, \ac{QAM} mapper, and \ac{OFDM} \ac{RG} mapper.

	For seamless integration of sensing into existing communication networks, the \ac{Tx} at the \ac{ISAC} transceiver is a standard \ac{MIMO}-\ac{OFDM} \ac{Tx} and, thus, the learnable parameter vector $\V \theta^{\rm t}$ of the \ac{Tx} at the \ac{ISAC} transceiver is void. Specifically, the \ac{DL}  \ac{ISAC} system $f_{\V \theta }$ can be rewritten as
	\begin{IEEEeqnarray}{rCl}
		\begin{pmatrix}
			\hat{\V b} \\ \hat{\V m}
		\end{pmatrix} = f_{\V \theta }\begin{pmatrix}
			\V b\\ \SceneColSet
		\end{pmatrix} 
		=
		\begin{pmatrix}
			g_{\V \theta _{\rm c}^{\rm r}} \circ h_{\rm c} \left(\SceneColSet, \cdot\right)\circ \tilde{g}\\
			g_{\V \theta _{\rm s}^{\rm r}}\left(\Tsr P, \cdot, \cdot\right)
			\circ
			\begin{pmatrix}
				\op I_{\rm d} \\ h_{\rm s}\left(\SceneColSet, \cdot\right) \circ \tilde{g}
			\end{pmatrix}
		\end{pmatrix}
		\circ g_{\rm c}^{\rm t}\left(\V b\right)\!,\nonumber\\
	\end{IEEEeqnarray}
	where the linear \ac{MIMO} precoder $\tilde{g}_{\V \theta^{\rm t}}$ degenerates to $\tilde{g}$ without any learnable parameter, and the corresponding precoding tensor is denoted by $\Tsr P$.
	In other words,
	the \ac{DL}  \ac{ISAC} system $f_{\V \theta }$ can be decomposed into two independently trainable subsystems $f^{\rm c}_{\V \theta _{\rm c}^{\rm r}}$ and $f^{\rm s}_{\V \theta _{\rm s}^{\rm r}}$, i.e.,
	\begin{IEEEeqnarray}{rCl}
		\hat{\V b} &=& f^{\rm c}_{\V \theta _{\rm c}^{\rm r}}\begin{pmatrix}
			\V b\\ \SceneColSet
		\end{pmatrix}
		= g_{\V \theta _{\rm c}^{\rm r}} \circ h_{\rm c}\left(\SceneColSet, \cdot\right) \circ \tilde{g} \circ g_{\rm c}^{\rm t}\left(\V b\right)\\
		\hat{\V m} &=& f^{\rm s}_{\V \theta _{\rm s}^{\rm r}}\begin{pmatrix}
			\V b\\ \SceneColSet
		\end{pmatrix} 
		=  g_{\V \theta _{\rm s}^{\rm r}} \left(\Tsr P, \cdot, \cdot\right) \circ
		\begin{pmatrix}
			\op I_{\rm d} \\ h_{\rm s}\left(\SceneColSet, \cdot \right) \circ \tilde{g}
		\end{pmatrix}
		\circ g_{\rm c}^{\rm t} \left(\V b\right)\,. \IEEEeqnarraynumspace
	\end{IEEEeqnarray}
	Therefore, the problem $\mathscr{P}_0$ can be decomposed into two subproblems $\mathscr{P}_1$ and $\mathscr{P}_2$, i.e.,
	\begin{IEEEeqnarray}{r'l'l}
		\mathscr{P}_1: &\mathop{\text{minimize}}\limits_{\V \theta_{\rm c}^{\rm r}  \in \Set \Theta_{\rm c}^{\rm r}}&
		\E{\frac{1}{\left|\Set Q\right|} \sum_{q\in\Set Q}L^{\rm comm}(\V b^{(q)}, \hat{\V b}^{(q)})}\\
		&\text{subject to}& \hat{\V b}^{(q)} = f_{\V \theta_{\rm c}^{\rm r} }^{\rm c}\begin{pmatrix}
			\V b^{(q)} \\ \SceneColSet^{(q)}
		\end{pmatrix}, \quad \forall q \in \Set Q
		\\
		\mathscr{P}_2: &\mathop{\text{minimize}}
		\limits_{\V \theta_{\rm s}^{\rm r}  \in \Set \Theta_{\rm s}^{\rm r}}
		&
		\E{\frac{1}{\left|\Set Q\right|}\sum_{q \in \Set Q}L^{\rm sens}(\Set  \Omega^{(q)}, \hat{\V m}^{(q)})}
		\\
		&\text{subject to}&  
		\hat{\V m}^{(q)} = f_{\V \theta _{\rm s}^{\rm r}}^{\rm s}\begin{pmatrix}
			\V b^{(q)}\\\SceneColSet^{(q)}
		\end{pmatrix}, \quad \forall q \in \Set Q\,.
	\end{IEEEeqnarray}
	\begin{remark}
		By removing the learnable parameter vector $\V \theta ^{\rm t}$ in the \ac{Tx} of the \ac{ISAC} transceiver, the trade-off coefficient $\eta$ in the general problem $\mathscr{P}_0$ can be eliminated, and the resulting subproblems $\mathscr{P}_1$ and $\mathscr{P}_2$ can be optimized separately. The subproblem $\mathscr{P}_1$ is a learning-based optimization problem for the neural \acp{Rx} of the \acp{UE}. 
		The subproblem $\mathscr{P}_2$ is a learning-based optimization problem for the neural sensing estimator at the \ac{ISAC} transceiver. 
		For a standard \ac{MIMO}-\ac{OFDM} network, since there are no learnable parameters in the communication link, the subproblem $\mathscr{P}_1$ does not need to be solved.
		Therefore, the solution to subproblem $\mathscr{P}_2$ enables the coexistence of sensing functionality with any standard \ac{MIMO}-\ac{OFDM} communication network.
	\end{remark}
	
	\subsection{Supervised Learning for Neural ISAC}\label{ssec:supvised_learning_nisac}
	Similar to the feature of the problem $\mathscr{P}_0$ mentioned in Remark~\ref{rmk:datapoint:input_eq_label}, for the problem $\mathscr{P}_2$, the collection $\SceneColSet$ that describes the scene serves as both the input and the label of the system $f_{\V \theta _{\rm s}^{\rm r}}^{\rm s}$. In practical implementations, one uses a map converter to map the collection $\SceneColSet$ into a discrete map representation as the label. We map the collection $\SceneColSet$ to the \acp{CIR} of both the communication and sensing channels with the NVIDIA Sionna platform \cite{HoyCamAouVemBinMarKel:22, HoyAouCamNimBinMarKel:23}.
	
	The Sionna scene for \ac{CIR} generation is first configured by the target positions generated by a random target setting generator, a Sionna scene seed including the material properties and shape information of the targets and the non-target environment,\footnote{In our implementation, the Sionna scene seed is generated by Blender and exported as an XML file.} and the antenna settings of the \ac{ISAC} transceiver and \acp{UE}. Since the Sionna platform limits that in one Sionna scene, the pattern of \acp{Rx} must be same, the \acp{CIR} for \acp{Rx} with different patterns cannot be generated in one Sionna scene. Therefore, for the case where the \ac{Rx} at the \ac{ISAC} transceiver and the \acp{Rx} of the \acp{UE} have different patterns, sensing and communication \acp{CIR} need to be generated in different Sionna scenes, i.e., sensing Sionna scene and communication Sionna scene. Then, the ray tracer generates paths for sensing and communication channels, and the solver computes the attenuations and delays along the generated paths. Finally, the discrete map label, sensing \acp{CIR}, and communication \acp{CIR} are collected as a datapoint.  
	
	The information bits for the communication task are generated by a binary source implicitly embedded into the neural \ac{ISAC} system.
	Then, the communication and sensing channels can be computed with the generated attenuations and delays.
	Therefore, the generated dataset $\tilde{\Set D}$ for practical implementations can be expressed equivalently as
	\begin{IEEEeqnarray*}{Cl}
		\tilde{\Set D} = \Big\{ \left(\V b^{(q)}, \Tsr H^{{\rm c}(q)}, \Tsr H^{{\rm s}(q)}, \V m^{(q)}\right): \hspace{-4 cm}& \\
		&
		\V b^{(q) } \in \{0,1\}^{\LengthInfoBit},\Tsr H^{{\rm c}(q)} \in \mathbb{C} ^{K \times N_{\rm t} \times W}, \IEEEyesnumber\label{eq:gen_dataset:eqv}\\
		&\Tsr H^{{\rm s}(q)} \in \mathbb{C} ^{N_{\rm r} \times N_{\rm t} \times W},\V m^{(q)} \in {[0,1]^{\times \left|\Set I\right|}}, q \in \Set Q\Big\}\!,
	\end{IEEEeqnarray*}
	where for the $q$th datapoint $\left(\V b^{(q)}, \Tsr H^{{\rm c}(q)}, \Tsr H^{{\rm s}(q)}, \V m^{(q)}\right)$, the discrete map representation $\V m^{(q)}$ serves as the label, and the remaining components $\V b^{(q)}$, $\Tsr H^{{\rm c}(q)}$, and $\Tsr H^{{\rm s}(q)}$ in the tuple are the inputs of the  neural \ac{ISAC} system.
	
	With the generated dataset $\tilde{\Set D}$ in \eqref{eq:gen_dataset:eqv}, the problem $\mathscr{P}_2$ can be rewritten as
	\begin{IEEEeqnarray}{rl'l}
		\mathscr{P}_3:\; &\mathop{\text{minimize}}
		\limits_{
			\substack{
				\V \theta_{\rm nn}  \in \Set \Theta_{\rm nn},\\
				\V \theta _{\rm f} \in \Set \Theta _{\rm f}
			}
		}
		&\E{
			\frac{1}{\left|\Set Q\right|}\sum_{q \in \Set Q}L^{\rm nn}(\V m^{(q)}, \hat{\V m}^{(q)})
		}
		\\
		&\text{subject to}&
		\hat{\V m}^{(q)} = g_{\V \theta _{\rm nn}}^{\rm nn}\left(\Tsr F^{(q)}\right)
		\label{eq:problem_3:condition:1}\\
		&&\Tsr F^{(q )}= g_{\V \theta _{\rm f}}^{\rm f} \left(\Tsr P^{(q)}, \Tsr S^{(q)}, \Tsr Y^{{\rm s}(q)}\right)
		\label{eq:problem_3:condition:2}\\
		&&\Tsr P^{(q)} = g_{\rm MIMO}\left({\Tsr H}^{{\rm c}{(q)}}\right)
		\label{eq:problem_3:condition:3}\\
		&&\Tsr Y^{{\rm s}(q)}= \Tsr H^{{\rm s}(q) }\times_3 \Tsr P^{(q)} \times_3 \Tsr S^{(q)} \nonumber\\
		&&\phantom{\text{$\Tsr Y^{{\rm s}(q)}=$}} +\> \Tsr Z^{{\rm s}(q)}
		\label{eq:problem_3:condition:4}\\
		&&\Tsr S^{(q)} = g_{\rm c}^{\rm t}(\V b^{(q)})
		\label{eq:problem_3:condition:5}
	\end{IEEEeqnarray}
	where 
	the learnable parameter vectors $\V \theta _{\rm nn} \in \Set \Theta_{\rm nn}$ and $\V \theta _{\rm f} \in \Set \Theta _{\rm f}$ partition the parameter vector $\V \theta _{\rm s}^{\rm r}$ of the problem $\mathscr{P}_2$;
	the function $L^{\rm nn}(\cdot, \cdot)$ is the loss function for \ac{NN} training, and its arguments are the discrete map representation of the scene $\SceneColSet$ and the estimated discrete map; 
	the \ac{NN} is denoted by $g_{\V \theta _{\rm nn}}^{\rm nn}$ and maps the extracted feature tensors to discrete maps;
	the feature extractor $g_{\V \theta _{\rm f}}^{\rm f}$ extracts effective features from the \ac{MIMO} precoding tensor, \ac{OFDM} \ac{RG} tensor, and the received sensing signal tensor;
	the mapping $g_{\rm MIMO}: \mathbb{C}^{K \times N_{\rm t} \times L} \rightarrow \mathbb{C}^{N_{\rm t} \times K \times L}$ maps the communication \ac{CSI} $\Tsr H^{{\rm c}(q)}$ to the \ac{MIMO} precoding tensor $\Tsr P^{(q)}$ to be applied before the sensing channel $\Tsr H^{{\rm s}(q)}$. 
	
	With different discrete map representations, the problem $\mathscr{P}_3$ can be converted to a multiclass or multilabel classification problem. 
	For the sensing task with probability map representations,  the labels $\V m^{(q)}$, $q \in \Set Q$ in the dataset are vectors where only one element is $1$, and the rest are $0$, where the elements of a label are in $[0,1]$, and their sum equals to $1$. These properties of the label vector can be used as constraints on the output vector of the neural \ac{ISAC} system to reduce the solution space.
	Therefore, the problem $\mathscr{P}_3$ can be converted into a multiclass classification problem $\mathscr{P}_3^{\rm mcc}$, given by
	\begin{IEEEeqnarray}{rC'l}
		\mathscr{P}_3^{\rm mcc}: &\mathop{\text{minimize}}
		\limits_{
			\substack{
				\V \theta_{\rm nn}  \in \Set \Theta_{\rm nn},\\
				\V \theta _{\rm f} \in \Set \Theta _{\rm f}
			}
		}
		&\E{
			\frac{1}{\left|\Set Q\right|}\sum_{q \in \Set Q}L_{\rm cce}(\V m^{(q)}, \hat{\V m}^{(q)})
		}
		\\
		&\text{subject to}&
		\text{conditions~\eqref{eq:problem_3:condition:1}--\eqref{eq:problem_3:condition:5}} \nonumber
		\\
		&& \sum_{i\in \Set I} \hat{m}^{(q)}_i = 1, \quad \forall q\in\Set Q\label{eq:p_3_mcc:sum}\\
		&&\hat{m}^{(q)}_i \in [0,1], \quad \forall q\in\Set Q,\, i \in \Set I\,.\label{eq:p_3_mcc:range}
	\end{IEEEeqnarray}
	Here, $L_{\rm cce}(\cdot, \cdot)$ is the \ac{CCE} loss, 
	the probability map labels $\V m^{(q)} \in \{0, 1\}^{\left|\Set I\right|}$, $q \in \Set Q$ are vectors in which only one entry is $1$ and the remaining entries are $0$. In the problem $\mathscr{P}_3^{\rm mcc}$, the predicted map can be interpreted as the probability of the target appearing in any of the given cells of the map grid.
	
	For the sensing task with hard or soft map representations, the labels $\V m^{(q)}$, $q \in \Set Q$ in the dataset are vectors whose elements are in $[0,1]$. This property of the label vector can be used as a constraint on the output vector of the neural \ac{ISAC} system to reduce the solution space.
	Therefore, the problem $\mathscr{P}_3$ can be converted into a multilabel classification problem,
	given by
	\begin{IEEEeqnarray}{rC'l}
		\mathscr{P}_3^{\rm mlc}: &\mathop{\text{minimize}}
		\limits_{
			\substack{
				\V \theta_{\rm nn}  \in \Set \Theta_{\rm nn},\\
				\V \theta _{\rm f} \in \Set \Theta _{\rm f}
			}
		}
		&\E{
			\frac{1}{\left|\Set Q\right|}\sum_{q \in \Set Q}L_{\rm bce}(\V m^{(q)}, \hat{\V m}^{(q)})
		}
		\\
		&\text{subject to}&
		\text{conditions~\eqref{eq:problem_3:condition:1}--\eqref{eq:problem_3:condition:5}}
		\\
		&& \hat{m}^{(q)}_i \in [0,1], \quad \forall q\in\Set Q,\, i \in \Set I\,.
	\end{IEEEeqnarray}
	Here, $L_{\rm bce}(\cdot, \cdot)$ is the \ac{BCE} loss.
	For the sensing task with hard map representations, the map representation labels $\V m^{(q)}$, $q \in \Set Q$ are vectors in the space $\{0,1\}^{\left|\Set I\right|}$. For the sensing task with soft map representations, the map representation labels $\V m^{(q)}$, $q \in \Set Q$ are vectors in the space $[0,1]^{\left|\Set I\right|}$. In the problem $\mathscr{P}_3^{\rm mlc}$, the predicted map can be interpreted as the target-to-cell volume or area ratio. 
	
	\begin{remark}\label{rmk:resolution} 
		Under discrete map representations, the classification problem type
		of the proposed scene reconstruction in the neural \ac{ISAC} system 
		is determined by both the number of targets and the map-grid cell size. 
		For example, in a single-target scene, if the map-grid cell size is smaller than the size of the target, the corresponding discrete map representation will have multiple non-zero entries, leading to the multilabel classification problem $\mathscr{P}_3^{\rm mlc}$. Conversely, in a multi-target scene, if the map-grid cell size is sufficiently large to encompass all targets within a single cell, the hard map representation will be reduced to the probability map representation, leading to the multiclass classification problem~$\mathscr{P}_3^{\rm mcc}$. 
	\end{remark}

	In the problem $\mathscr{P}_3$ and its variants $\mathscr{P}_3^{\rm mcc}$ and $\mathscr{P}_3^{\rm mlc}$, the neural sensing estimator $g_{\V \theta _{\rm s}^{\rm r}}$ is decomposed into 
	two subsystems, namely, the feature extractor $g_{\V \theta _{\rm f}}^{\rm f}$ and NN $g_{\V \theta _{\rm nn}}^{\rm nn}$. 
	The proposed signal processing flow of the neural sensing estimator first extracts features from wireless signals and precoding tensors, and then feeds them into the \ac{NN} for scene~reconstruction.
	
	\section{Signal Processing Flow of the Neural Sensing Estimator}\label{sec:sp_flow}
	
	\begin{figure*}[!t]
		\centering
		\includegraphics[width=0.9\textwidth]{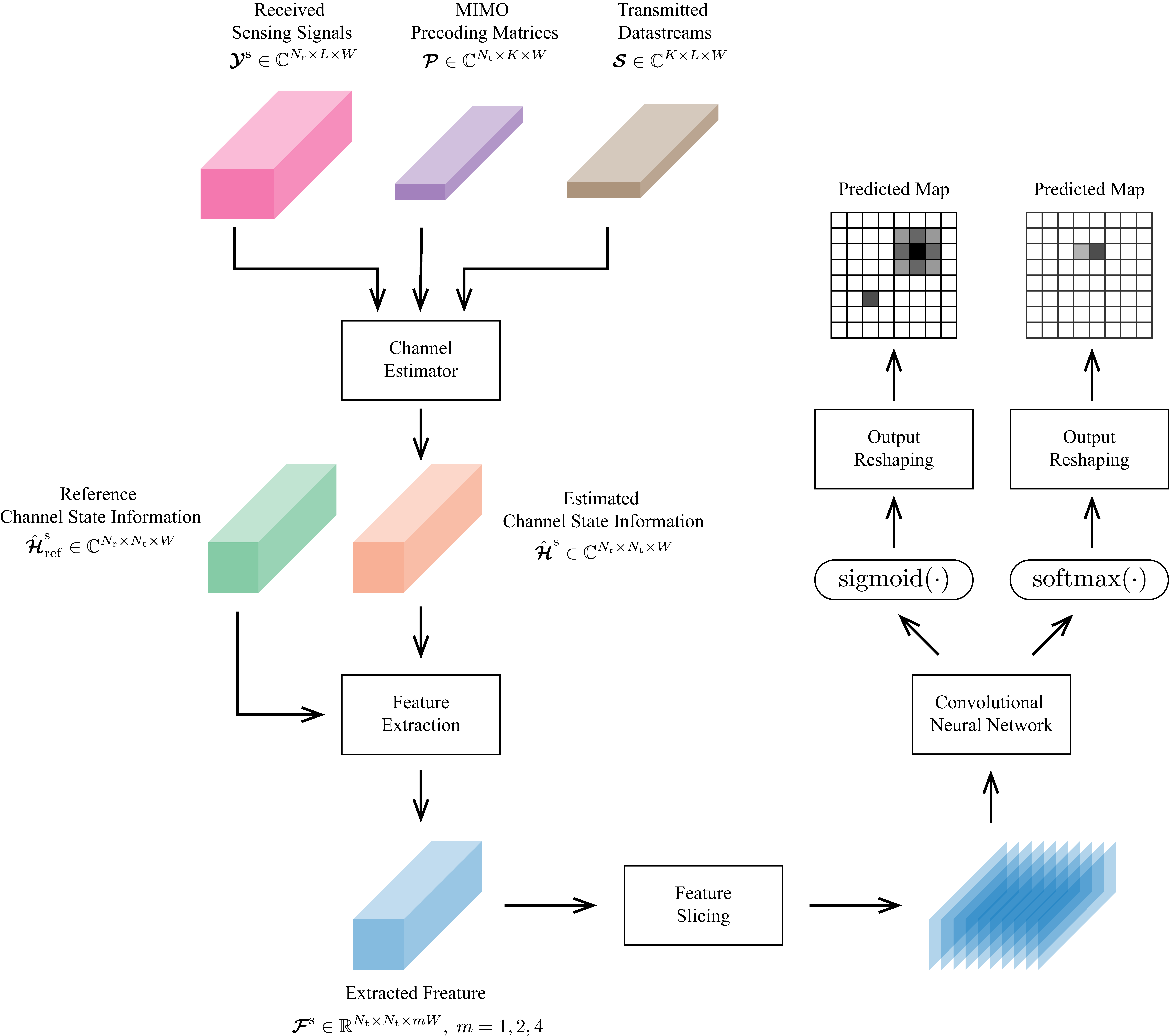}
		\caption{The signal processing flowchart of the proposed neural sensing estimator.}
		\label{fig:arch_neural_sensing}
	\end{figure*}
	
	This section presents the proposed signal processing flow of the neural sensing estimator, 
	which specifies the structure of the system $g_{\V \theta _{\rm nn}}^{\rm nn} \circ g_{\V \theta _{\rm f}}^{\rm f}$.
	As shown in Fig.~\ref{fig:arch_neural_sensing}, the sensing \ac{CSI} is first estimated from the \ac{MIMO} precoding tensor, received sensing signal tensor, and the transmitted \ac{OFDM} \ac{RG} tensor. Then, the effective features for \ac{NN} training are extracted from the estimated sensing \ac{CSI} and reference sensing \ac{CSI}. Finally, the extracted feature tensor is sliced to match the input shape of the \ac{NN}.

	\subsection{Sensing Channel Estimation}\label{ssec:sens_chest}
	From the signal model of \ac{DL} sensing with \ac{MU}-\ac{MIMO} in \eqref{eq:signal_model_sens:tensor}, it can be observed that the transmitted datastream tensor $\Tsr S$ and the noise tensor $\Tsr Z^{\rm s}$ are irrelevant to the scene information.
	In contrast, the \ac{MIMO} precoding tensor $\Tsr P$ derived from the estimated communication \ac{CSI} contains scene information of the communication channels.
	Nevertheless, it is important to note that the multipaths of the communication channels overlap with those of the sensing channels.  This is because any path radiated from the \ac{ISAC} transceiver can be scattered back through diffuse reflection. Therefore, the sensing \ac{CSI} contains effective information for scene reconstruction \cite{YanZhoLiu:13}. 
	
	\subsubsection{Problems in the Sensing Channel Estimation}
	
	The sensing channel estimation for the $w$th subcarrier can be formulated as 
	\begin{IEEEeqnarray}{r'l'l}
		\mathscr{P}_w^{\rm sce}: &\mathop{\text{minimize}}
		\limits_{\mathring{\M H}^{\rm s}_w  \in \mathbb{C}^{N_{\rm r}\times N_{\rm t}}}
		&{\left\Vert \M Y^{\rm s}_w - \mathring{\M H}^{\rm s}_w \M P_w \M S_w \right\Vert_{\rm F}}\,.
	\end{IEEEeqnarray}
	Different from communication channel estimation, the estimation of sensing channels in \ac{ISAC} systems faces two additional challenges: 
	(i) for each subcarrier, the rank of the transmitted datastream matrix is limited by the number of datastreams for communication, which may lead to a rank deficiency; and
	(ii) the attenuation in the sensing channel is more severe than that in the direct path of the communication channel, leading to a lower \ac{SNR} in received sensing signals.
	
	To further elaborate on the challenge (i), we examine \eqref{eq:sys_model:sens:w}, where it can be observed that
	\begin{IEEEeqnarray}{rCl}
		r_w\triangleq\mathrm{rank}\left(\M P_w \M S_w\right)  &\leq& \min\left\{\mathrm{rank}\left(\M P_w\right), \, \mathrm{rank}\left(\M S_w \right)\right\} \nonumber\\
		& \leq &K, \qquad \forall w \in \Set W\,.
	\end{IEEEeqnarray}
	\ac{MU}-\ac{MIMO} communication networks require that the number of datastreams is not more than the number of \ac{Tx} antennas for \ac{DL} communication, i.e., $K \leq N_{\rm t}$. Furthermore, if the number of datastreams is strictly smaller than the number of \ac{Tx} antennas, i.e., $K < N_{\rm t}$, then 
	\begin{IEEEeqnarray}{rCl}
		r_w \leq K < N_{\rm t}, \qquad \forall w \in \Set W\,.
	\end{IEEEeqnarray}
	In this case, the rank-deficiency in $\M P_w \M S_w$ changes \ac{LS} estimation in the sensing channel estimation. For the case where $N_{\rm t} \leq L$, since $\M P_w \M S_w \M S_w^\nH \M P_w^\nH \in \mathbb{C} ^{N_{\rm t} \times N_{\rm t}}$ is non-invertible, its right pseudo-inverse $\left(\M P_w \M S_w\right)^\dagger$ cannot be solely expressed by $\M P_w \M S_w$ as
	\begin{IEEEeqnarray}{rCl}
		\M S^\nH_w\M P^\nH_w \left(\M P_w \M S_w \M S_w^\nH \M P_w^\nH\right)^{-1}\,.
	\end{IEEEeqnarray}
	Instead, this requires a \ac{SVD} to obtain the pseudo-inverse. 
	Furthermore, the \ac{LS} estimation turns out to be biased, as summarized in Proposition~\ref{prop:biased_lse}.
	\begin{proposition}[Biased \ac{LS} estimation]\label{prop:biased_lse}
		The expectation of the \ac{LS} estimation is given by
		\begin{IEEEeqnarray}{rCl}
			\E{\hat{\M H}_w^{\rm s}} = \M H_w^{\rm s}\sum_{k=1}^{r_w}\left[\M U_w\right]_{:, k}\left[\M U_w\right]_{:, k}^\nH\,,
		\end{IEEEeqnarray}
		where the unitary matrix
		$\M U_w \in \mathbb{C} ^{N_{\rm t} \times N_{\rm t}}$ is defined by the \ac{SVD} of the matrix $\M P_w\M S_w$, i.e.,
		\begin{IEEEeqnarray}{rCl}
			\M P_w \M S_w = \M U_w \M \Sigma_w \M V_w^\nH
		\end{IEEEeqnarray} 
		in which the matrix $\M \Sigma_w \in \mathbb{R} ^{N_{\rm t} \times L}$ is a rectangular diagonal matrix with the first $r_w$ elements on the main diagonal being positive real numbers, and the remaining elements being zero;
		and the matrix $\M V_w \in \mathbb{C} ^{L \times L}$ is unitary.
	\end{proposition}
	\begin{proof}
		See Appendix~\ref{apd:proof:prop:biased_lse}
	\end{proof}
	\begin{remark}
		In this case, the \ac{LS} estimation is a biased estimation problem in which a projection operation onto the subspace spanned by the $r_w$ eigenvectors $[\M U_w]_{:,k}$, $k= 1, 2, \dotsc, r_w$, is introduced. 
		In fact, it selects the solution with the smallest Frobenius-norm among the infinite solutions to the underdetermined problem~$\mathscr{P}_w^{\rm sce}$.
	\end{remark}

	\subsubsection{Tikhonov Regularization and Implicit Noise Averaging}
	To address the rank-deficiency issue, a method based on matrix completion is proposed in \cite{LiuLiuLiMasEld:J22} where the matrix $\M P_w \M S_w$ is completed to be full-rank. However, the cost of the matrix completion is the modification of the \ac{MIMO} precoder in the communication links. To solve this challenge under the framework of standard-compatible sensing integration to the existing communication networks, we introduce Tikhonov regularization into the \ac{LS} estimation of full-rank cases. Specifically, the result to this estimation problem is given by
	\begin{IEEEeqnarray}{rCl}
		\hat{\M H}_w^{\rm s}  = \M Y_w^{\rm s} \M S_w^\nH \M P_w^\nH\left(\M P_w \M S_w \M S_w^\nH \M P_w^\nH + \lambda _{\rm reg} \M I_{N_{\rm t}}\right)^{-1}\,,\label{eq:ls_est:tikhonov}
	\end{IEEEeqnarray}
	where the scalar $\lambda _{\rm reg}$ is the regularization factor. In fact, Tikhonov regularization imposes an additional smoothness constraint on the problem $\mathscr{P}_w^{{\rm sce}}$, and it is summarized as Proposition~\ref{prop:tikhonov}.
	
	\begin{proposition}[Tikhonov regularization]\label{prop:tikhonov}
		The estimation \eqref{eq:ls_est:tikhonov} is the solution of the problem  $\mathscr{P}_w^{{\rm sce}(1)}$ formulated as
		\begin{IEEEeqnarray}{rl'l}
			\mathscr{P}_w^{{\rm sce}(1)}: & 
			\mathop{\text{minimize}}_{\mathring{\M H}_w^{\rm s} \in \mathbb{C} ^{N_{\rm r} \times N_{\rm t}}} &
			\left\Vert \M Y^{\rm s}_w - \mathring{\M H}^{\rm s}_w \M P_w \M S_w \right\Vert^2_{\rm F}+\lambda _{\rm reg} \left\Vert \mathring{\M H}_w^{\rm s} \right\Vert_{\rm F}^2\,.\nonumber\\
		\end{IEEEeqnarray}
	\end{proposition}
	\begin{proof}
		See Appendix~\ref{apd:proof:prop:tikhonov}.
	\end{proof}
	
	\begin{remark}\label{rmk:prop:tikhonov}
		By introducing the regularization parameter $\lambda _{\rm reg}$, the solution space is relaxed from the minimum Frobenius-norm constraint.
		Therefore, by adjusting the regularization factor, the ratio of the estimated sensing \ac{CSI} Frobenius norm to the noise Frobenius norm can be tuned to extract more information from the measurements.
	\end{remark} 
	
	To show the implicit noise averaging for challenge (ii), a conventional approximation \cite{LiuLiuLiMasEld:J22,LuLiuDonXioXuLiuJin:J24} on the transmitted datastream matrix $\M S_w$ is adopted, i.e.,
	\begin{IEEEeqnarray}{rCl}
		\M S_w \M S_w^\nH \approx L \M I_{K},\qquad \forall w \in \Set W\,. \label{eq:clt:approx}
	\end{IEEEeqnarray}
	By the central limit theorem, this approximation is valid when $L \gg K$.
	
	With the approximation \eqref{eq:clt:approx}, the estimation in \eqref{eq:ls_est:tikhonov} can be simplified to
	\begin{IEEEeqnarray}{rCl}
		\hat{\M H}_w^{\rm s} = \frac{1}{L}\M Y_w^{\rm s} \M S_w^\nH \M P_w^\nH \left(\M P_w \M P_w^\nH + \frac{\lambda _{\rm reg}}{L}\M I_{N_{\rm t}}\right)^{-1}\,.
	\end{IEEEeqnarray}
	It can be shown that the random matrix $\M Y_w^{\rm s} \M S_w^\nH/L$ implicitly averages out the noise in measurements.
	Substituting \eqref{eq:clt:approx} into \eqref{eq:sys_model:sens:w}, we have
	\begin{IEEEeqnarray}{rCl}
		\frac{1}{L} \M Y_w^{\rm s} \M S_w^\nH = \M H_{w}^{\rm s} \M P_w + \frac{1}{L} \M N_w^{\rm s} \M S_w^\nH\,.
	\end{IEEEeqnarray}
	To show the noise in measurements is averaged out, we need to show that the variance of the entries of $ \M N_w^{\rm s} \M S_w^\nH/L$ decreases as $L$ increases, which is summarized as Lemma~\ref{prop:noise_avg}.
	
	\begin{lemma}[Implicit noise averaging]\label{prop:noise_avg}
		If $L \gg K$, the entries of $ \M N_w^{\rm s} \M S_w^\nH/L$  follow independently from a zero-mean complex Gaussian distribution with a variance of $N_0/L$.
	\end{lemma}
	\begin{proof}
		See Appendix~\ref{apd:proof:prop:noise_avg}.
	\end{proof}
	Furthermore, the equivalence between more measurements and lower noise variance can be implied from Lemma~\ref{prop:noise_avg}, which is summarized as Proposition~\ref{coro:noise_avg}.
	
	\begin{proposition}[Equivalence between more measurements and lower noise]\label{coro:noise_avg}
		If $L \gg K$ and $B$ is a positive integer, then the entries of the random matrices $\frac{1}{B L} \tilde{\M Y}_w^{\rm s} \tilde{\M S}_w^\nH$ and  $\frac{1}{L} \M Y_w^{\rm s} \M S_w^\nH $, respectively, given by
		\begin{IEEEeqnarray}{rCl}
			\frac{1}{B L} \tilde{\M Y}_w^{\rm s} \tilde{\M S}_w^\nH 
			&=&
			\M H_w^{\rm s} \M P_w + \frac{1}{BL} \tilde{\M N}_{w}^{\rm s} \tilde{\M S}_w^\nH 
			\\
			\frac{1}{L} \M Y_w^{\rm s} \M S_w^\nH 
			&=&
			\M H_w^{\rm s} \M P_w + \frac{1}{L} \mathring{\M N}_w^{\rm s} \M S_w^\nH\,,
		\end{IEEEeqnarray}
		are independent and follow the same distribution. Here, 
		$\tilde{\M Y}_w^{\rm s}, \tilde{\M N}_w^{\rm s} \in \mathbb{C} ^{N_{\rm r} \times BL}$, $\tilde{\M S}_w \in \mathbb{C} ^{K\times BL} $, and $\mathring{\M N}_w^{\rm s} \in \mathbb{C} ^{K \times L}$;
		the entries of $\tilde{\M N}_w^{\rm s}$ follow independently from a zero-mean complex Gaussian distribution with a variance of $N_0$; and
		the entries of $\mathring{\M N}_w^{\rm s}$ follow independently from a zero-mean complex Gaussian distribution with a variance of $N_0/B$.
	\end{proposition}
	\begin{proof}
		By Lemma~\ref{prop:noise_avg}, the variance of the entries of $ \frac{1}{B L} \tilde{\M N}_w^{\rm s} \tilde{\M S}_w^\nH $ is $N_0/\left(BL\right)$, and the variance of  the entries of $ \frac{1}{ L} \mathring{\M N}_w^{\rm s} {\M S}_w^\nH $ is $N_0/(BL)$.
	\end{proof}
	
	\begin{remark}
		Before the estimated sensing \ac{CSI} tensor is fed into the feature extractor, normalization can be applied to adjust the dynamic range of the estimation. For example, the estimated \ac{CSI} can be normalized by the modulus of its element with the largest modulus $\max\{|\tilde{\Tsr H}^{\rm s}|\}$, given by
		\begin{IEEEeqnarray}{rCl}
			\hat{\Tsr H}^{\rm s}= \tilde{\Tsr H}^{\rm s}/\max\{|\tilde{\Tsr H}^{\rm s}|\}\,.
		\end{IEEEeqnarray}
		where
		\begin{IEEEeqnarray}{rCl}
			\tilde{\Tsr H}^{\rm s} = \begin{bmatrix}
				\hat{\M H}^{\rm s}_{1} & \hat{\M H}^{\rm s}_{2} & \cdots & \hat{\M H}^{\rm s}_W
			\end{bmatrix}_3\,.
		\end{IEEEeqnarray}
	\end{remark}
	
	\subsection{Feature Extraction from Sensing \ac{CSI}}\label{ssec:feature_ext}
	To train the \ac{NN} for scene reconstruction, it is critical to extract effective features from the estimated sensing \ac{CSI}.
	
	\subsubsection{Beamspace and Delay-Domain Features}\label{ssec:beam_delay_features}
	To extract the angular information of targets from the estimated sensing \ac{CSI}, a \ac{DFT} is applied to the first and second dimensions of the estimated sensing \ac{CSI} tensor $\hat{\Tsr H}^{\rm s} \in \mathbb{C} ^{N_{\rm r} \times N_{\rm t} \times W}$.
	Specifically, the extracted feature tensor ${\Tsr F}_{\rm beam}^{\rm s}$ in beamspace is given by
	\begin{IEEEeqnarray}{rCl}
		{\Tsr F}_{\rm beam}^{\rm s} &=&
		\M F_{N_{\rm t}} \times_2 \M F_{N_{\rm r}} \times_1 \hat{\Tsr H}^{\rm s}\,.
	\end{IEEEeqnarray}
	For the $w$th subcarrier, the spectrum $\big|\left[\Tsr F_{\rm beam}^{\rm s}\right]_{:,:,w}\big|$  of extracted beamspace feature $\left[\Tsr F_{\rm beam}^{\rm s}\right]_{:,:,w}$ reflects the \ac{Tx} and \ac{Rx} beam indices of the targets. 
	
	To extract the range information of targets from the estimated sensing \ac{CSI}, an \ac{IDFT} is applied to the third dimension of the tensor $\hat{\Tsr H}^{\rm s}$. The extracted feature $\Tsr F_{\rm delay}^{\rm s}$ is given by
	\begin{IEEEeqnarray}{rCl}
		\Tsr F_{\rm delay}^{\rm s} = \M F_{W}^\nH \times_3 \hat{\Tsr H}^{\rm s}\,.
	\end{IEEEeqnarray}
	Consider the independence among the three dimensions of the tensor $\hat{\Tsr H}^{\rm s}$. The feature with angular and range information of targets is given by
	\begin{IEEEeqnarray}{rCl}
		\Tsr F_{\rm beam, delay}^{\rm s} = \M F_{W}^\nH \times_3 \Tsr F_{\rm beam}^{\rm s}.
	\end{IEEEeqnarray}
	To avoid any errors induced by the phase mismatch between the \ac{Tx} and \ac{Rx} of the \ac{ISAC} transceiver, the spectrum of the features in beamspace and delay domain is adopted as the input features, given by
	\begin{IEEEeqnarray}{rCl}
		{\Tsr F}_{\rm bd}^{\rm s} = \left|\Tsr F_{\rm beam, delay}^{\rm s}\right| \in \mathbb{R} ^{N_{\rm r} \times N_{\rm t} \times W}\,. \IEEEyesnumber \label{eq:feature:beam_delay}
	\end{IEEEeqnarray}

	\subsubsection{Direct Features} \label{ssec:direct_features}
	To validate the effectiveness of the extracted feature $\Tsr F^{\rm s}$, a real-valued tensor $\Tsr F^{\rm s}_{\rm d} \in \mathbb{R} ^{N_{\rm r} \times N_{\rm t} \times 2W}$ formed by concatenating the real and imaginary parts of the normalized sensing \ac{CSI} estimation $\hat{\Tsr H}^{\rm s}$ is adopted as a baseline. Specifically, the tensor $\Tsr F_{\rm d}^{\rm s}$ of the direct features is given by
	\begin{IEEEeqnarray}{rCl}
		\Tsr F^{\rm s}_{\rm d} &=&\Big[
		\Re\left\{\left[\hat{\Tsr H}^{\rm s}\right]_{:,:, 1}\right\} \; \cdots \; \Re\left\{\left[\hat{\Tsr H}^{\rm s}\right]_{:,:, W}\right\}\nonumber \\
		&& \phantom{\text{$\Big[$}}
		\Im\left\{\left[\hat{\Tsr H}^{\rm s}\right]_{:,:, 1}\right\}  \;\cdots \; \Im\left\{\left[\hat{\Tsr H}^{\rm s}\right]_{:,:, W}\right\} 
		\Big]_3\,. 
	\end{IEEEeqnarray}

	\subsubsection{Feature Fusion with Prior Knowledge on Environments}
	If prior knowledge on the static environment, including self-interference between the \ac{Tx} and \ac{Rx} antennas at the \ac{ISAC} transceiver,  is available, then features can be constructed by fusing the feature of prior knowledge and the features extracted from the estimated sensing \ac{CSI} $\hat{\Tsr H}^{\rm s} $. In particular, the feature extraction is performed as described in Section~\ref{ssec:beam_delay_features} or \ref{ssec:direct_features}.
	The prior knowledge on the static environment in the form of reference sensing \ac{CSI} $\hat{\Tsr H}_{\rm ref}^{\rm s}$ can be acquired during a calibration phase by sending unitary training pilots to the target-free environment. In practice, the calibration can be performed during nighttime when the targets are not in the \ac{RoI}. 
	To match the estimated sensing \ac{CSI} $\hat{\Tsr H}^{\rm s}$ influenced by the \ac{MIMO} precoder and transmitted datastreams, features are extracted from the precoding-imposed reference sensing \ac{CSI}, given by
	\begin{IEEEeqnarray}{rCl}\label{eq:def:precoding_imposed_ref_CSI}
		\left[\hat{\Tsr H}_{\rm p, ref}^{\rm s}\right]_{:,:,w}&=& \left[\hat{\Tsr H}^{\rm s}_{\rm ref}\right]_{:,:,w} \M P_w \M S_w \M S_w^\nH \M P_w^\nH\\
		&&\times\> \left(\M P_w \M S_w \M S_w^\nH \M P_w^\nH  + \lambda _{\rm reg} \M I_{N_{\rm t}}\right)^{-1},\quad \forall w\in\Set W\,.\nonumber
	\end{IEEEeqnarray}
	Here, \eqref{eq:def:precoding_imposed_ref_CSI} is obtained by substituting $\M Y_w = \left[\hat{\Tsr H}_{\rm ref}^{\rm s}\right]_{:,:,w} \M P_w \M S_w$ into \eqref{eq:ls_est:tikhonov}.
	Subsequently, the feature tensor $\Tsr F_{\rm e, ref}^{\rm s}$ is extracted from the precoding-imposed reference sensing \ac{CSI} $\hat{\Tsr H}_{\rm p, ref}^{\rm s}$ 
	in the same manner as
	the feature tensor $\Tsr F_{\rm e}^{\rm s}$ from the estimated sensing \ac{CSI}.
	
	The fused feature tensor $\Tsr F^{\rm s}$ is given by
	\begin{IEEEeqnarray}{rCl}
		\Tsr F^{\rm s} = g_{\rm fus}\left(\Tsr F^{\rm s}_{\rm e}, \Tsr F^{\rm s}_{\rm e, ref}\right)\!,
	\end{IEEEeqnarray}
	where the fusion mapping $g_{\rm fus}$ of the features $\Tsr F_{\rm e}^{\rm s}$ and $ \Tsr F_{\rm e, ref}^{\rm s}$ can be chosen as subtraction, stacking, or reference-free operations, respectively, defined as follows:
	\begin{IEEEeqnarray}{r'rCl}
		g^{\rm SUB}_{\rm fus}: &(\Tsr F_1, \Tsr F_2)& \mapsto &\Tsr F_1 - \Tsr F_2
		\\
		g^{\rm STA}_{\rm fus}: &(\Tsr F_1, \Tsr F_2) & \mapsto& \begin{bmatrix}
			\Tsr F_1 & \Tsr F_2
		\end{bmatrix}_2
		\\
		g^{\rm NOR}_{\rm fus}:& (\Tsr F_1, \Tsr F_2)& \mapsto &\Tsr F_1\,.
	\end{IEEEeqnarray}

	\subsection{Feature Slicing}
	To match the input of the NN, the feature tensor needs to be reshaped before being fed into the NN. 
	For \acp{MLP}, reshaping the feature tensor is trivial, as the tensor is reshaped into a vector. 
	However, for \acp{CNN}, the way the tensor is sliced can affect the performance of the trained networks. We decided to slice along the second dimension of the feature tensor, and the feature map associated with the $n_{\rm t}$th input channel of the \ac{CNN} is given by $\left[\Tsr F^{\rm s}\right]_{:,n_{\rm t}, :} $.
	
	\subsection{Neural Network Training}
	
	\subsubsection{Network Structures}\label{ssec:network_structure}
	Although we do not impose any constraints on the network structure under the proposed neural \ac{ISAC} framework, we choose \acp{CNN} due to their performance in extracting structural and spatial information from feature maps. To efficiently train deep networks, we devise the network using residual connections based on ResNet18 \cite{HeZhaRenSun:C16}. 
	The backbone of the \ac{CNN} includes $8$ cascaded residual blocks. 
	Moreover, before the residual blocks, an input convolutional layer and max pooling layer are applied to initially aggregate information from the sliced feature maps. Different from the typical ResNet18, we replace the $7\times7$ convolutional kernel with a $3\times 3$ convolution due to the smaller size of the input feature maps. After the residual blocks, there is an average pooling layer followed by a fully connected layer. As shown in Fig.~\ref{fig:arch_neural_sensing}, for hard or soft map estimation, the output activation function of the fully connected layer is sigmoid. For probability map estimation, the output activation function of the fully connected layer is softmax.
	
	\subsubsection{Training}
	\begin{algorithm}[!t]
		\caption{\ac{NN} Training for the Neural Sensing Estimator}
		\label{alg:nn_training}
		\begin{algorithmic}[1]
			\STATE \textbf{Input:} Training dataset $\tilde{\Set D}_{\rm train}$, validation dataset $\tilde{\Set D}_{\rm val}$, Tikhonov regularization factor $\lambda_{\rm reg}$, CNN model $g\sb{\V \theta_{\rm nn} }^{\rm nn}(\cdot)$, noise variance $N_0$, initial learning rate $\eta_{\rm init}$
			\STATE \textbf{Output:} The trained parameters $\mathring{\V \theta}_{\rm nn}$ of the CNN
			\STATE Initialize the parameters $\V \theta_{\rm nn}$ of CNN $g\sb{\V \theta_{\rm nn} }^{\rm nn}(\cdot)$
			\STATE Initialize learning rate $\eta \leftarrow \eta_{\rm init}$
			\STATE Initialize best validation loss $\mathring{L}_{\rm val} \leftarrow \infty$
			
			\FOR{$i = 1$ to $N\sb{{\rm epoch}}$}
			\FOR{$(\V b^{(q)}, \Tsr H^{{\rm c}(q)}, \Tsr H^{{\rm s}(q)}, \V m^{(q)})$ in $\tilde{\Set D}_{\rm train}$}
			\STATE \# ISAC Tx:
			\STATE $\Tsr S \leftarrow g_{\rm c}^{\rm t}(\V b^{(q)})$ \# Transmitted OFDM RG
			\STATE $\Tsr P^{\rm ZF} \leftarrow g_{\rm ZF}(\Tsr H^{{\rm c}(q)})$  \# Zero-forcing precoding
			
			\STATE  \# Sensing channel propagation:
			\STATE $\Tsr Z^{\rm s} \leftarrow \text{ComplexGaussian}(N_0)$  \# Noise generation
			\STATE $\Tsr Y^{{\rm s}} \leftarrow \Tsr H^{{\rm s}(q)} \times_3 \Tsr P^{\rm ZF} \times_3 \Tsr S + \Tsr Z^{\rm s}$
			
			\STATE  \# Sensing channel estimation:
			\STATE $\hat{\Tsr H}^{\rm s} \leftarrow \text{ChannelEstimator}(\Tsr Y^{\rm s}, \Tsr P^{\rm ZF}, \Tsr S, \lambda _{\rm reg})$
			
			\STATE \# Input features for sensing:
			\STATE $ \Tsr F^{\rm in} \leftarrow \text{FeatureExtractor}(\hat{\Tsr H}^{\rm s}, \Tsr H^{\mathrm{s}(q)}, \Tsr P^{\rm ZF}, \Tsr S, \lambda _{\rm reg})$
			
			\STATE \# Update learnable parameters of CNN:
			\STATE $L \leftarrow L_{\rm nn}(\V m^{(q)},  g^{\rm nn}\sb{\V \theta_{\rm nn} }(\Tsr F^{\rm in}))$
			\STATE $\V \theta_{\rm nn} \leftarrow \text{AdamOptimizer}(\V \theta_{\rm nn}, \eta, \nabla_{\V \theta _{\rm nn}} L)$
			\ENDFOR
			\STATE \# Validation
			\STATE $L_{\rm val} \leftarrow \text{ModelEvaulation}(\tilde{\Set D}_{\rm val}, \V \theta _{\rm nn})$
			\IF{$L_{\rm val} < \mathring{L}_{\rm val}$}
			\STATE $\mathring{L}_{\rm val}, \mathring{\V \theta }_{\rm nn} \leftarrow L_{\rm val}, \V \theta _{\rm nn}$
			\ENDIF
			\ENDFOR
			\RETURN $\mathring{\V \theta}_{\rm nn}$
		\end{algorithmic}
	\end{algorithm}
	
	The \ac{NN} training algorithm for the proposed neural sensing estimator is presented in Algorithm~\ref{alg:nn_training}. 
	The training algorithm first initializes the learnable parameter of the \ac{CNN} model $g_{\V \theta _{\rm nn}}^{\rm nn}(\cdot)$ and the learning rate $\eta$. 
	Without loss of generality, and for the sake of clarity in presentation, Algorithm~\ref{alg:nn_training} considers the case where the batch size is $1$. 
	For the $q$th batch, the information bits $\V b^{(q)}$ to be transmitted are mapped to transmitted \ac{OFDM} \ac{RG} $\Tsr S$ through the subsystem $g_{\rm c}^{\rm t}\left(\cdot\right)$, as shown in \eqref{eq:problem_3:condition:3}--\eqref{eq:problem_3:condition:5}. 
	The communication channel in frequency domain is mapped to a zero-forcing precoding tensor $\Tsr P^{\rm ZF}$ through the mapping $g_{\rm ZF}(\cdot)$ whose operation on the dimensions of \ac{Tx} antenna index and datastream index is specified as 
	\begin{IEEEeqnarray}{rCl}
		\left[g_{\rm ZF}(\hat{\Tsr H}^{\rm c})\right]_{:,:,w} \triangleq \big(\big[\hat{\Tsr H}^{\rm c}\big]_{:,:,w}\big)^\dagger \M \Xi_{w}, \quad w \in \Set W\,,
	\end{IEEEeqnarray}
	where $\hat{\Tsr H}^{\rm c} \in \mathbb{C} ^{N_{\rm t} \times K \times W}$ is the estimated communication \ac{CSI} from the \ac{UL} transmission,\footnote{If the perfect \ac{CSI} condition is assumed, then $\hat{\Tsr{H}}^{\rm c}=\Tsr{H}^{\rm c}$.} and the diagonal matrix $\M \Xi_w \in \mathbb{C} ^{K\times K}$ normalizes $\|[\hat{\Tsr H}^{\rm c}]_{:,k,w}\|_2$, $k = 1, 2, \dotsc, K$, to $1$.
	The noise generator $\text{ComplexGaussian}(N_0)$ generates a noise tensor $\Tsr Z^{\rm s}$, where the elements are independent, identically distributed and drawn from a zero-mean complex Gaussian distribution with variance $N_0$.
	According to the signal model \eqref{eq:signal_model_sens:tensor}, the received sensing signal tensor $\Tsr Y^{\rm s}$ is generated from the sensing channel tensor $\Tsr H^{{\rm s}(q)}$ along with the previously generated tensors $\Tsr S$, $\Tsr P^{\rm ZF}$, and $\Tsr Z^{\rm s}$.
	With tensors $\Tsr Y^{\rm s}$, $\Tsr P^{\rm ZF}$, and $\Tsr S$ as the input of the neural sensing estimator, the forward inference follows the signal processing flow as shown in Fig.~\ref{fig:arch_neural_sensing}. Specifically, the function $\text{ChannelEstimator}(\Tsr Y^{\rm s}, \Tsr P^{\rm ZF}, \Tsr S, \lambda _{\rm reg})$ provides the estimated sensing \ac{CSI} tensor $\hat{\Tsr H}^{\rm s}$ from the input tensors of the neural sensing estimator according to Section~\ref{ssec:sens_chest}. Then, the function $\text{FeatureExtractor}(\hat{\Tsr H}^{\rm s}, \Tsr H^{\mathrm{s}(q)}, \Tsr P^{\rm ZF}, \Tsr S, \lambda _{\rm reg})$ extracts and slices the fused feature tensor $\Tsr F^{\rm s}$ according to Section~\ref{ssec:feature_ext}. With the sliced input features, the learnable parameters of the \ac{CNN} model $g_{\V \theta _{\rm nn}}^{\rm nn}(\cdot )$ as detailed in Section~\ref{ssec:network_structure} are updated with an Adam optimizer \cite{Kin:14}.
	The algorithm returns the trained model parameters $\mathring{\V \theta }_{\rm nn}$ with the lowest validation loss.

	\section{Case Studies}
	\label{sec:casestudy}
	This section covers simulations under the proposed signal processing framework with different extracted features and feature fusion techniques. First, an example is provided to illustrate scene reconstruction with a neural sensing estimator under the three discrete map representations. Then, the performance of different neural \ac{ISAC} systems within the proposed signal processing framework is evaluated across various map resolutions, Tikhonov regularization factors, environmental conditions, target speeds, and bandwidths. In addition, the primary limitation of the proposed framework is discussed.
	
	\subsection{Evaluation Setup}
	\begin{figure}[tp]
		\centering
		\includegraphics[width=0.95\columnwidth]{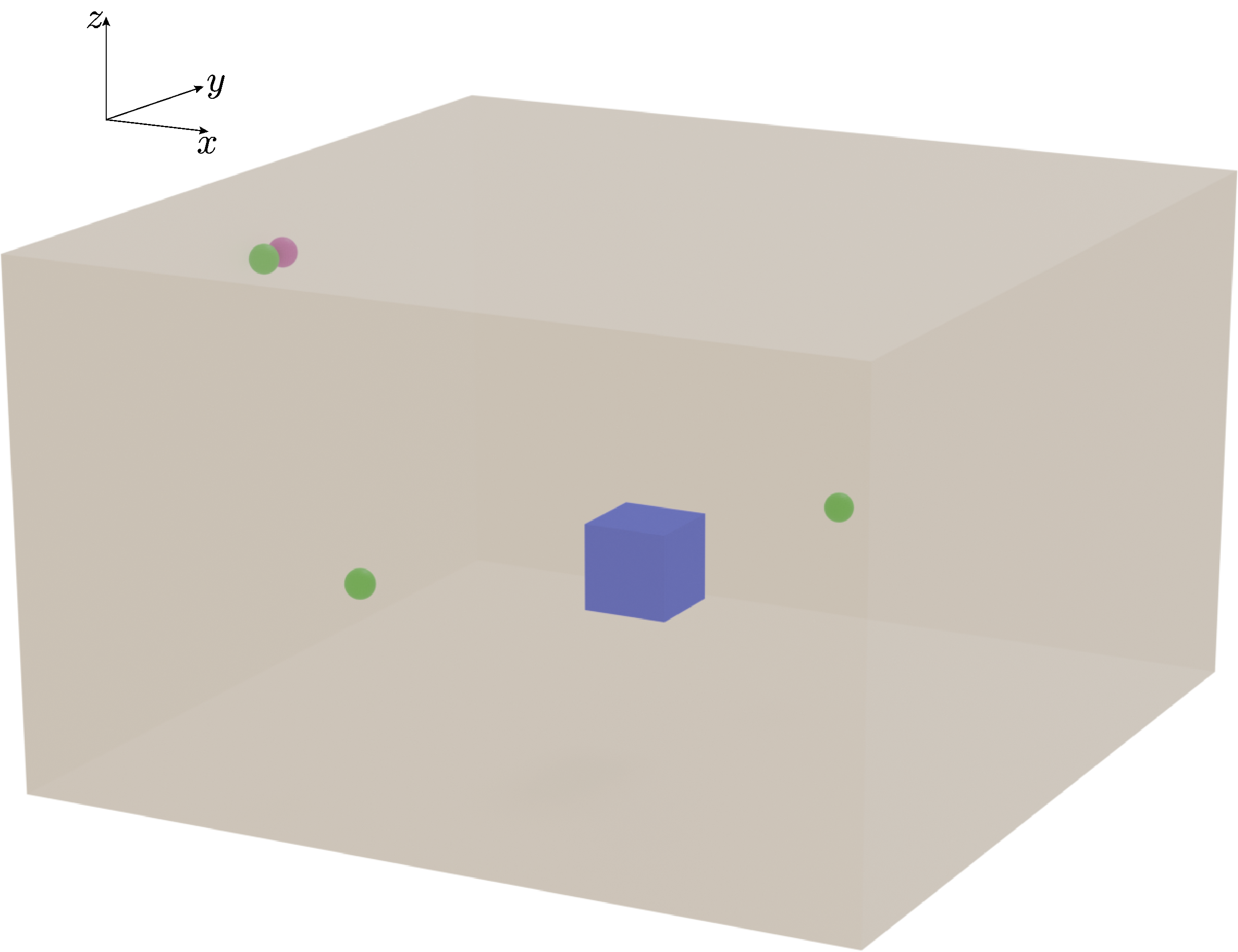}
		\caption{The considered scene in our case studies. The blue cube is a target, and the two small green balls close to the target are \acp{Rx} at \acp{UE}. The pair of small pink balls and small green balls corresponds to the \ac{Tx} and \ac{Rx} at the \ac{ISAC} transceiver, respectively.}
		\label{fig:sceneplot}
	\end{figure}
	
	\newcommand{\subfigwidth}{0.32\textwidth}
	\newcommand{\subfigfigwidth}{0.7\textwidth}
	\newcommand{\subfigfigwidthSpecial}{0.63\textwidth}
	\begin{figure*}[!t]
		\begin{subfigure}[b]{\subfigwidth}
			\centering
			\includegraphics[width=\subfigfigwidth]{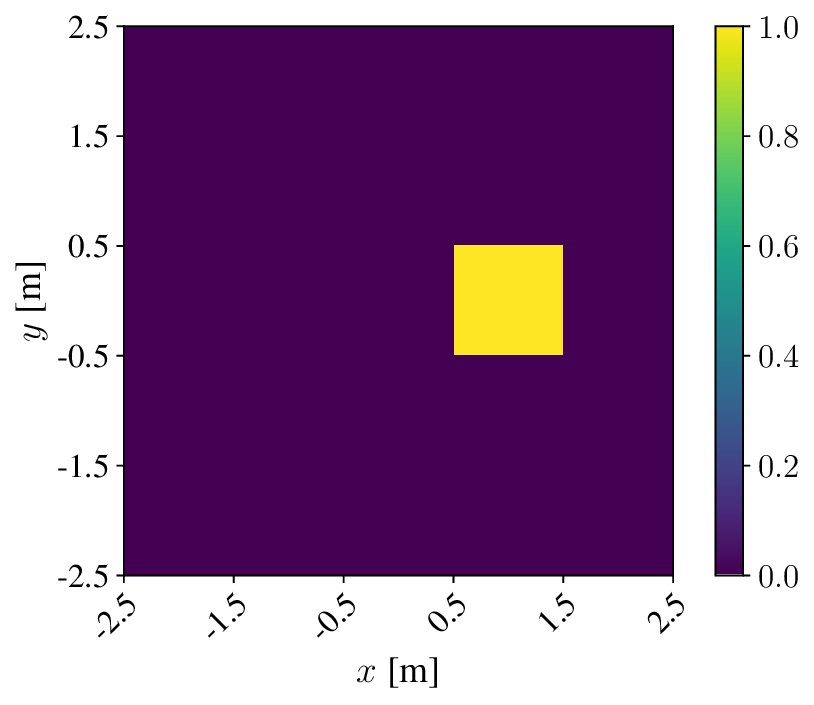}
			\caption{}
			\label{fig:examples:prob_map:label}
		\end{subfigure}
		\hfill
		\begin{subfigure}[b]{\subfigwidth}
			\centering
			\includegraphics[width=\subfigfigwidth]{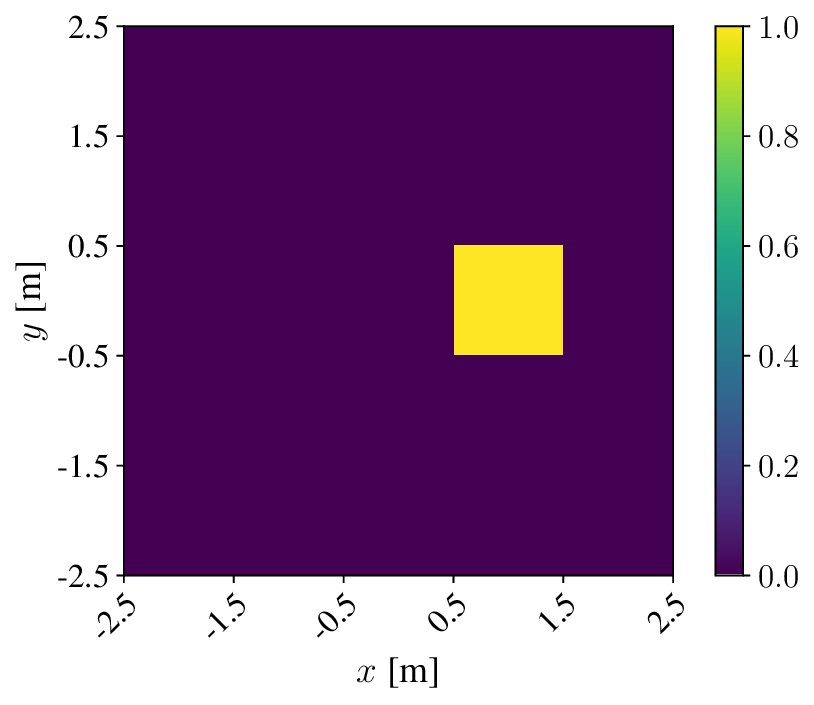}
			\caption{}
			\label{fig:examples:prob_map:hard_predict}
		\end{subfigure}
		\hfill
		\begin{subfigure}[b]{\subfigwidth}
			\centering
			\includegraphics[width=\subfigfigwidth]{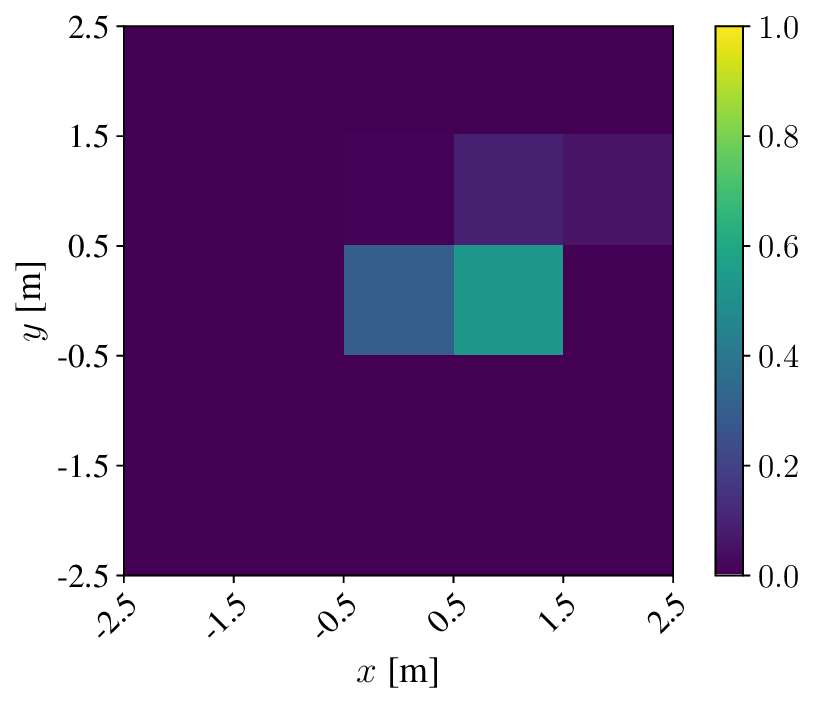}
			\caption{}
			\label{fig:examples:prob_map:soft_predict}
		\end{subfigure}
		\vspace{1em}
		\begin{subfigure}[b]{\subfigwidth}
			\centering
			\includegraphics[width=\subfigfigwidth]{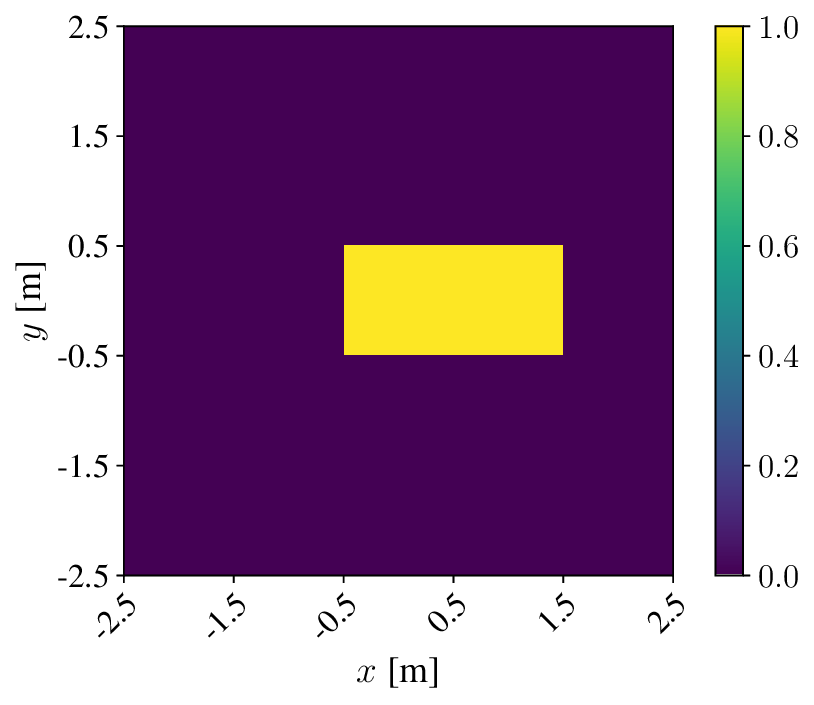}
			\caption{}
			\label{fig:examples:hard_map:label}
		\end{subfigure}
		\hfill
		\begin{subfigure}[b]{\subfigwidth}
			\centering
			\includegraphics[width=\subfigfigwidth]{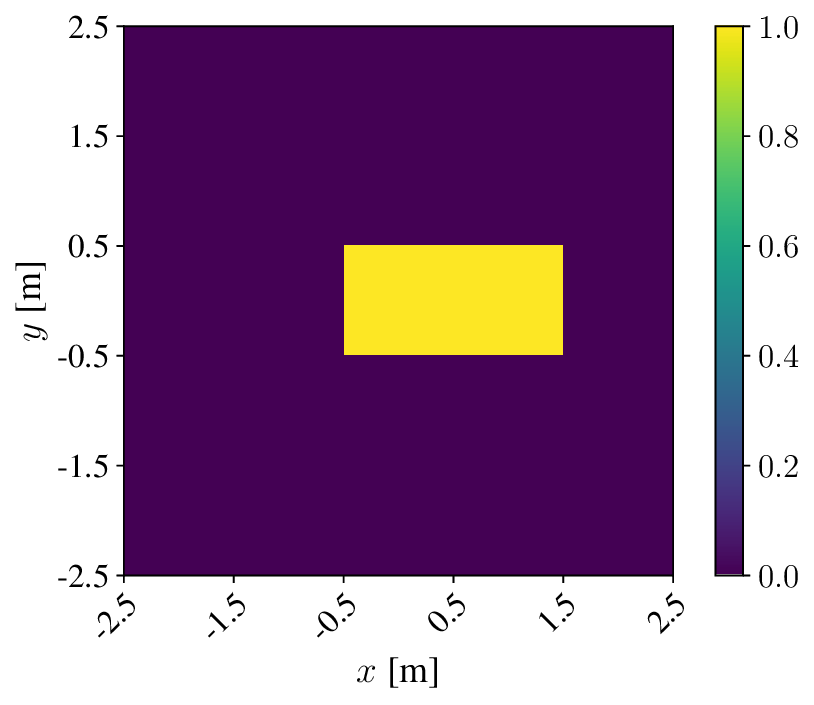}
			\caption{}
			\label{fig:examples:hard_map:hard_predict}
		\end{subfigure}
		\hfill
		\begin{subfigure}[b]{\subfigwidth}
			\centering
			\includegraphics[width=\subfigfigwidth]{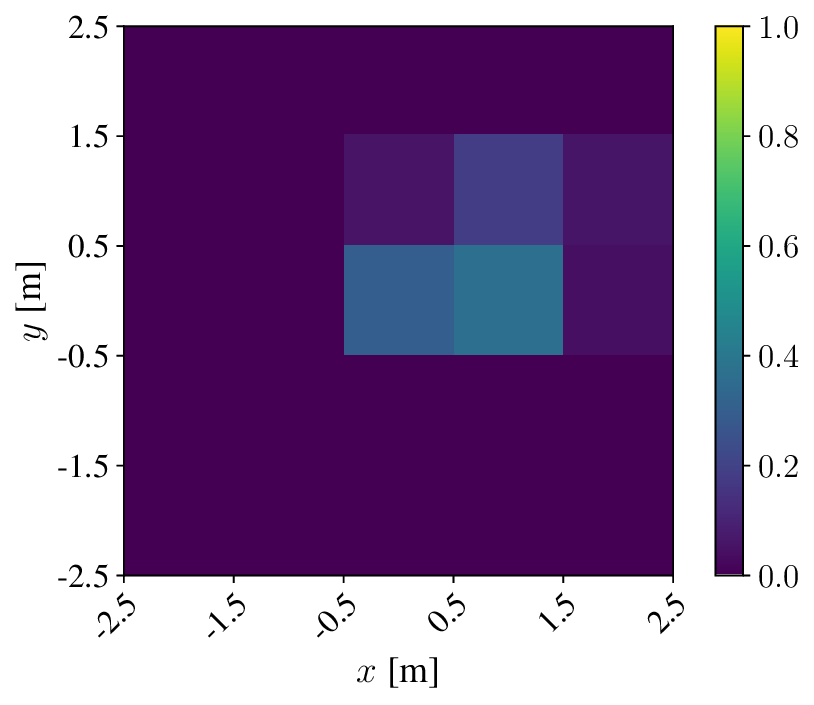}
			\caption{}
			\label{fig:examples:hard_map:soft_predict}
		\end{subfigure}
		\vspace{1em}
		\begin{subfigure}[b]{\subfigwidth}
			\centering
			\includegraphics[width=\subfigfigwidthSpecial]{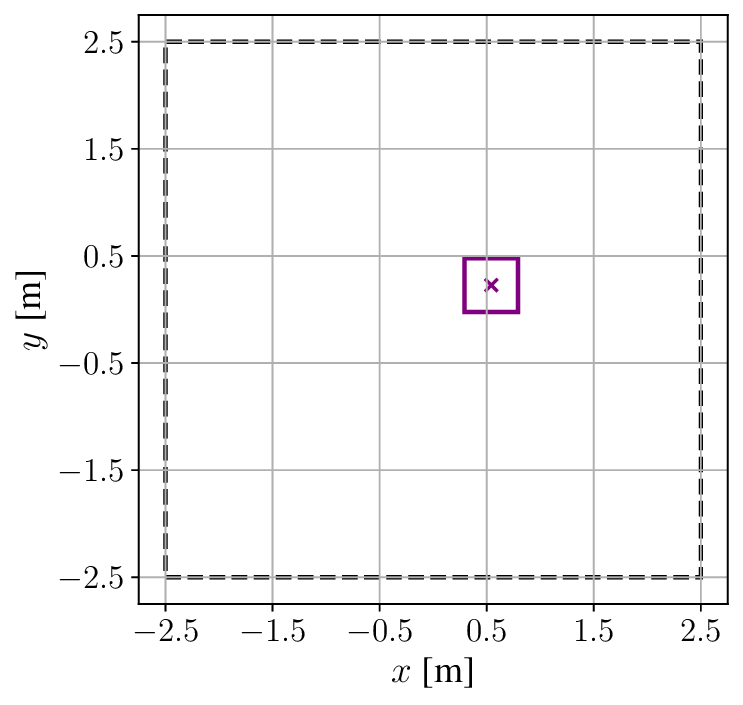}
			\caption{}
			\label{fig:examples:groud_truth}
		\end{subfigure}
		\hfill
		\begin{subfigure}[b]{\subfigwidth}
			\centering
			\includegraphics[width=\subfigfigwidth]{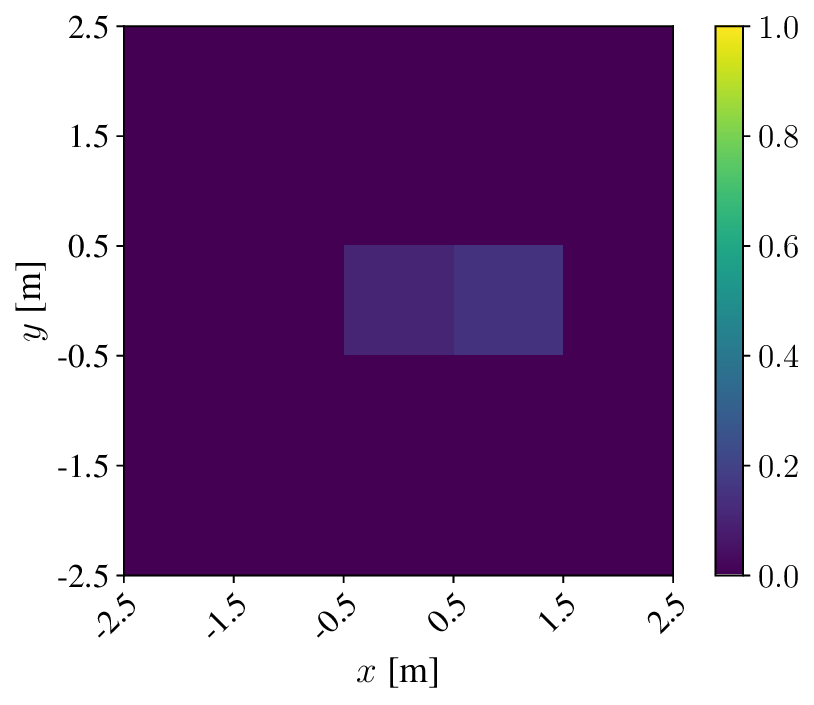}
			\caption{}
			\label{fig:examples:soft_map:label}
		\end{subfigure}
		\hfill
		\begin{subfigure}[b]{\subfigwidth}
			\centering
			\includegraphics[width=\subfigfigwidth]{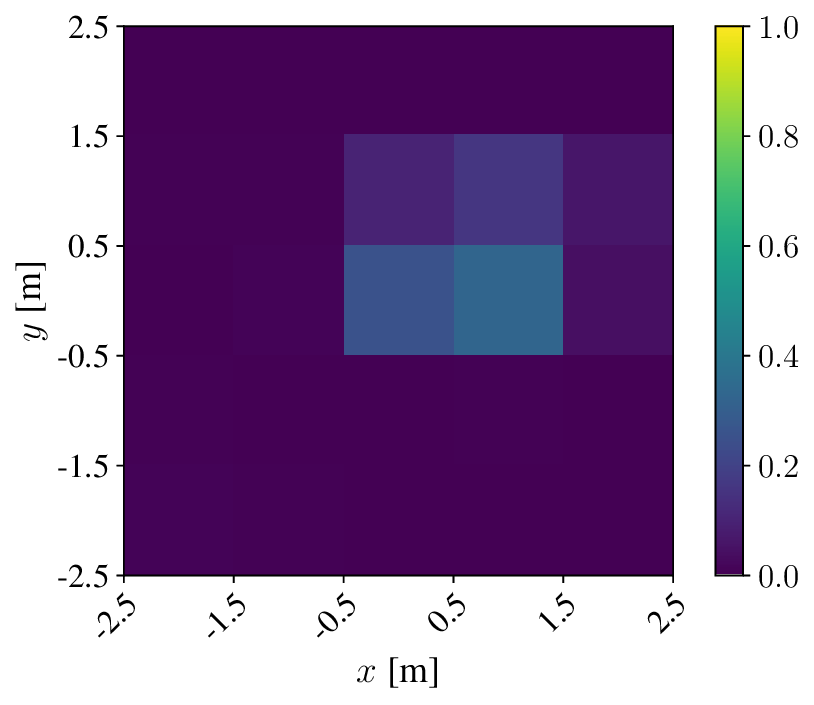}
			\caption{}
			\label{fig:examples:soft_map:soft_predict}
		\end{subfigure}
		\caption{An example of scene reconstruction under probability map, hard map, and soft map representations with the proposed framework.  
			(a) probability map label 
			(b) one-hot encoding of the probability map prediction
			(c) probability map prediction
			(d) hard map label 
			(e) hardened hard map prediction with a threshold of $0.2$ 
			(f) hard map prediction
			(g) ground truth 
			(h) soft map label 
			(i) soft map prediction 
		}
		\label{fig:examples}
	\end{figure*}
	For the case studies, NVIDIA Sionna \cite{HoyCamAouVemBinMarKel:22, HoyAouCamNimBinMarKel:23} is used to generate datasets, and to train and evaluate the neural ISAC systems.
	The considered \ac{ISAC} systems consist of one \ac{ISAC} transceiver with $8$ \ac{Tx} antennas and $8$ \ac{Rx} antennas and two \acp{UE}, each equipped with one \ac{Rx} antenna. 
	We set the center frequency as $6$\,GHz \cite{NaiParAshLeh:J20}.
	One \ac{OFDM} \ac{RB} has $14$ \ac{OFDM} symbols, with the $3$rd and $12$th \ac{OFDM} symbols being pilots arranged in a Kronecker pattern.
	The number of subcarriers is set to $128$, and the default bandwidth is set to $40$\,MHz.
	The first $5$ and the last $6$ subcarriers are guard subcarriers, and the \ac{DC}  subcarrier is set to null.
	The transmitted power at the \ac{ISAC} transceiver is set to $10$\,mW.
	The antennas are vertically polarized.
	The pattern of the antennas at the \ac{ISAC} transceiver is set to the \ac{BS} pattern in the 3GPP 38.901 standard \cite{3GPP:38901}. 
	The pattern of the antennas at the \acp{UE} is set to be isotropic.
	The \ac{Tx} and \ac{Rx} antennas at the \ac{ISAC} transceiver are arranged as \acp{ULA} with a spacing of half operating wavelength, and the \acp{ULA} are oriented toward the origin of the coordinate system. 
	The self-interference cancellation amount between the \ac{Tx} and \ac{Rx} antennas at the \ac{ISAC} transceiver is set to $-40$\,dB.
	The centers of the \ac{Tx} antennas and the \ac{Rx} antennas at the \ac{ISAC} transceiver are set to $(-2.4, 0.1, 2.5)$\,m and $(-2.4, -0.1, 2.5)$\,m, respectively. 
	The communication \ac{CSI} is assumed to be known from the \ac{UL} transmissions.
	The targets are cubes with a Lambertian scattering pattern and a scattering coefficient of $1$.
	The noise figure is set to $9$\,dB and antenna temperature is set to $290$\,K \cite{NgoAshYanLarMar:J17}.
	
	In the following simulations, the scene shown in Fig.~\ref{fig:sceneplot} is considered for the case studies. 
	The \ac{ISAC} transceivers, two \acp{UE}, and one target are enclosed in a $5$\,m$\times5$\,m$\times3$\,m concrete box whose geometric center is placed at $(0, 0, 1.5)$\,m.
	The two \acp{UE}, represented by the two small green balls at the lower position in Fig~\ref{fig:sceneplot}, are uniformly distributed in the space $[-2, 2] \times [-2, 2] \times 1$\,m$^3$.
	The target is a $0.5$\,m$\times 0.5$\,m$\times 0.5$\,m cube, represented by the blue cube in Fig~\ref{fig:sceneplot}. Its geometric center is  uniformly placed in the space $[-2, 2]\times [-2, 2] \times \{1\}$\,m$^3$. The \ac{RoI} $\Set R$ is the space $[-2.5, 2.5]\times [-2.5, 2.5] \times \{1\}$\,m$^3$. The targets and the \acp{UE} have no physical overlap.
	By default, the map resolution of the three representations is set to $5\times 5$.
	The observation interval is set to the duration corresponding to $20$ \ac{OFDM} \acp{RB}.
	
	The dataset is generated as described in Section~\ref{ssec:supvised_learning_nisac}.
	The number of paths for \ac{CIR} generation between one pair of \ac{Tx} and \ac{Rx} antennas is set to $75$. The number of path samples for the Sionna ray tracer is set to $10^5$, and the maximal depth is set to $5$. The reflection, scattering, and diffraction effects are considered in \ac{CIR} generation.
	The number of epochs is set to $100$.
	The default batch size is set to $1024$.
	The initial learning rate is set to $10^{-5}$, and the Adam optimizer is used for parameter updates.
	The size of the generated dataset is $10^5$.
	The training and test datasets are split from the generated dataset, and the train-test ratio is set to $4:1$.

	\subsection{An Example of Probability Map, Hard Map, and Soft Map Estimations}
	To illustrate the proposed discrete map estimation, we exhibit an example of probability map, hard map, and soft map estimations.
	The sensing channel estimation is performed with Tikhonov regularization, as specified in \eqref{eq:ls_est:tikhonov}.
	The regularization factors are different for scene reconstruction under different map representations. Specifically, the regularization factors for probability map, hard map, and soft map estimations are set to $10^{-2}$, $10^{-3}$, and $10^{-3}$, respectively.
	The features extracted from the channel estimates are the beamspace and delay-domain features specified in Section~\ref{ssec:beam_delay_features}.
	The reference sensing \ac{CSI} is assumed to be known.
	The features extracted from the estimated sensing \ac{CSI} and reference sensing \ac{CSI} are fused via the subtraction operation~$g_{\rm fus}^{\rm SUB}$.
	
	Fig.~\ref{fig:examples} shows an example of scene reconstruction under the discrete map representations associated to a sample in the test dataset. 
	Fig.~\ref{fig:examples:groud_truth} shows the \ac{2D} ground truth of the scene to be reconstructed. The \ac{RoI} is indicated by the black dashed lines, and the gray vertical and horizontal lines partition the RoI into $1$\,m$\times 1$\,m cells. The target is indicated by the purple box, and its geometric center is indicated by the purple cross. 
	Fig.~\ref{fig:examples:prob_map:label},~\ref{fig:examples:hard_map:label},~and~\ref{fig:examples:soft_map:label} are the probability map label, hard map label, and soft map label, respectively. The goal of scene reconstruction under the three representations is to reconstruct these labels.
	Fig.~\ref{fig:examples:prob_map:soft_predict},~\ref{fig:examples:hard_map:soft_predict},~and~\ref{fig:examples:soft_map:soft_predict} are the direct outputs of the output activation $\text{sigmoid}(\cdot) $ or $\text{softmax}(\cdot)$ as shown in Fig.~\ref{fig:arch_neural_sensing}.
	In addition, to match the probability map label and hard map label, the direct map predictions of the neural sensing estimator need to be discretized to $\{0,1\}$-valued maps. Specifically, the estimated probability map in Fig~\ref{fig:examples:prob_map:hard_predict} is the one-hot encoding of the probability map prediction in Fig.~\ref{fig:examples:prob_map:soft_predict}, and the estimated hard map in Fig.~\ref{fig:examples:hard_map:hard_predict} is obtained by discretizing the hard map prediction in Fig.~\ref{fig:examples:hard_map:soft_predict} with a threshold of $0.2$.
	In this example, the trained neural sensing estimators for probability map and hard map correctly estimate the corresponding map labels, while the trained neural sensing estimator for soft map cannot perfectly match the prediction with the label. 
	This is due to the additional robustness provided by the discretizations in the probability map and hard map estimations.
	The comparison among the predictions in Fig.~\ref{fig:examples:prob_map:soft_predict},~\ref{fig:examples:hard_map:soft_predict},~and~\ref{fig:examples:soft_map:soft_predict} shows that the probability value of the cell occupied by the geometric center of the target is the lowest (darkest in color) in the soft map prediction. This indicates that training with soft map labels enables the trained neural sensing estimator to learn partial shape information of the target.

	\subsection{Performance Comparison Among Different Feature Fusions}
	
	In this simulation, we focus on neural \ac{ISAC} systems with beamspace and delay-domain features as described in Section~\ref{ssec:beam_delay_features}. 
	Three feature fusions and two channel estimators, as mentioned in Section~\ref{sec:sp_flow}, are evaluated and compared. 
	Since the considered \ac{RoI} is a square, square cells are always considered. Therefore, the number of cells per side can be used to represent the map resolution.
	For scene reconstruction under the probability map representation, the performance metric is the accuracy in correctly identifying whether the cell is occupied by the geometric center of the target.
	For scene reconstruction under the hard map representation, the performance metric is the precision-recall break-even point. Specifically, this metric is obtained by first sweeping the threshold for discretizing the hard map prediction to a $\{0,1\}$-valued map and then interpolating to acquire a precision-recall curve and then picking the point where precision equals to recall.
	For scene reconstruction under the soft map representation, the performance metric is the \ac{BCE} loss, which reflects the distance between label and prediction.
	
	\begin{figure}[!t]
		\centering
		\includegraphics[width=0.9\columnwidth]{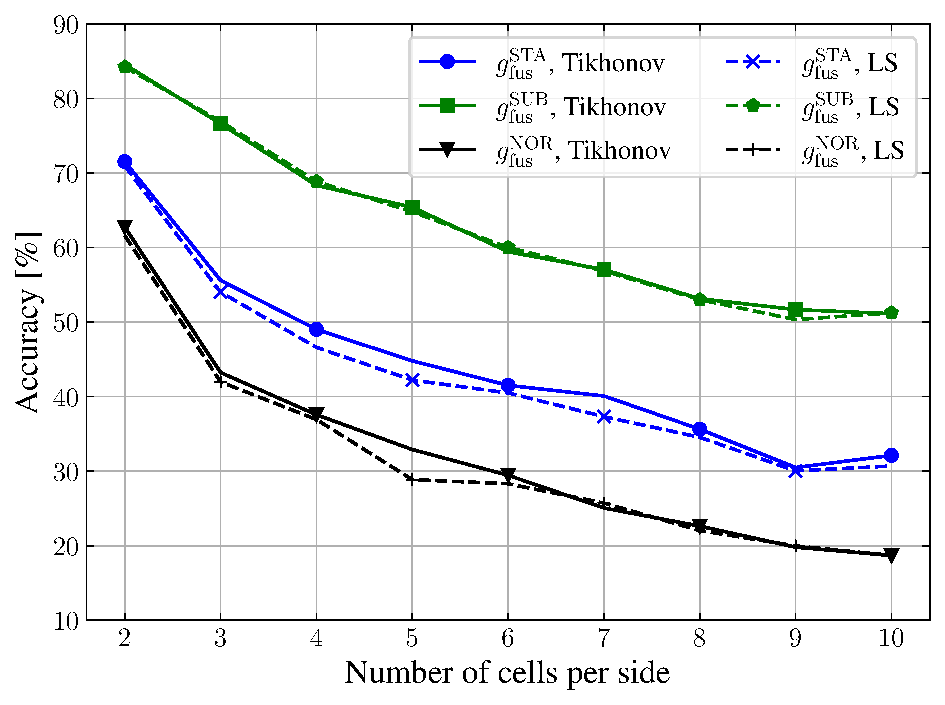} 
		\caption{Performance of probability map estimations with respect to the number of cells per side.}
		\label{fig:probmap_fusion_compare}
	\end{figure}
	
	\begin{figure}[!t]
		\centering
		\includegraphics[width=0.9\columnwidth]{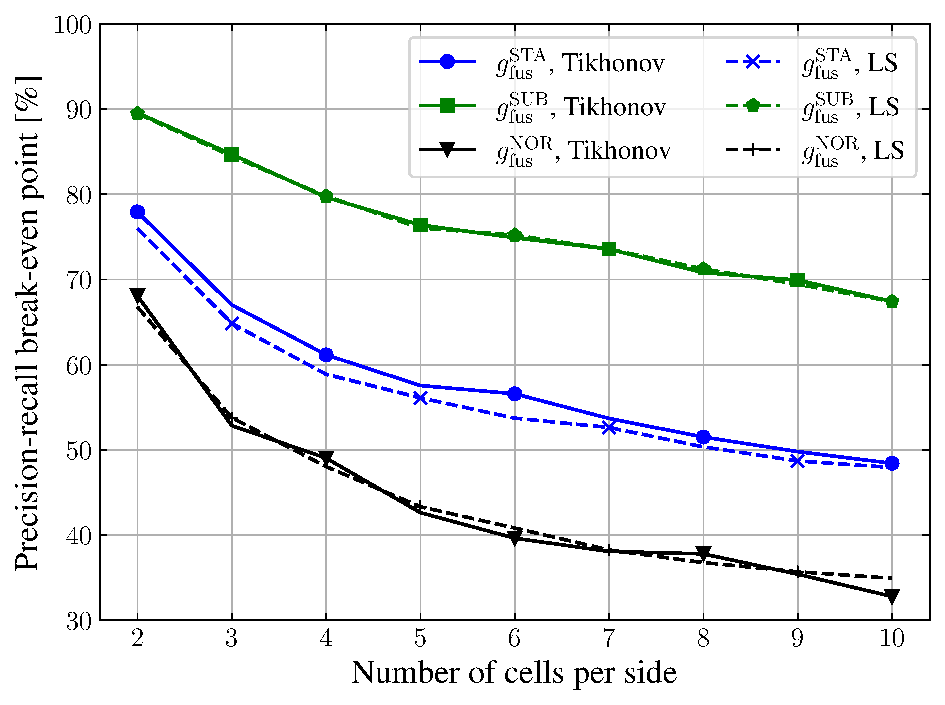}
		\caption{Performance of hard map estimations with respect to the number of cells per side.}
		\label{fig:hardmap_fusion_compare}
	\end{figure}
	
	The Figures~\ref{fig:probmap_fusion_compare}~and~\ref{fig:hardmap_fusion_compare} show that as the number of cells per side increases, the accuracy and precision-recall break-even point decrease. This is because a higher number of cells in a map grid results in smaller cell sizes, leading to a higher likelihood of incorrect decisions.
	Under a fixed map resolution, the accuracy and the precision-recall break-even point reflect the estimation quality of the reconstructed maps. Therefore, this indicates a fundamental trade-off between map resolution and estimation quality in scene reconstruction. Specifically, given limited information, increasing map resolution constrains the upper bound of the achievable estimation quality.
	The Figures~\ref{fig:probmap_fusion_compare}~and~\ref{fig:hardmap_fusion_compare} also show that, for both types of channel estimation considered, neural \ac{ISAC} systems using beamspace and delay-domain features achieve the best performance with subtraction fusion, followed by stacking fusion, and lastly by the case without any environmental prior knowledge, denoted as $g_{\rm fus}^{\rm SUB}$, $g_{\rm fus}^{\rm STA}$, and $g_{\rm fus}^{\rm NOR}$, respectively. 
	In the probability map estimation, better performance refers to higher accuracy, whereas in the hard map estimation, it refers to a higher precision-recall break-even point. 
	The superior performance of fusions $g_{\rm fus}^{\rm SUB}$ and $g_{\rm fus}^{\rm STA}$ compared to $g_{\rm fus}^{\rm NOR}$ is due to their incorporation of additional environmental prior information. Moreover, the superiority of fusion $g_{\rm fus}^{\rm SUB}$ over $g_{\rm fus}^{\rm STA}$ indicates that the subtraction fusion provides more effective prior information for environmental component removal compared to directly training with stacked features.

	\begin{figure}[!t]
		\centering
		\includegraphics[width=0.9\columnwidth]{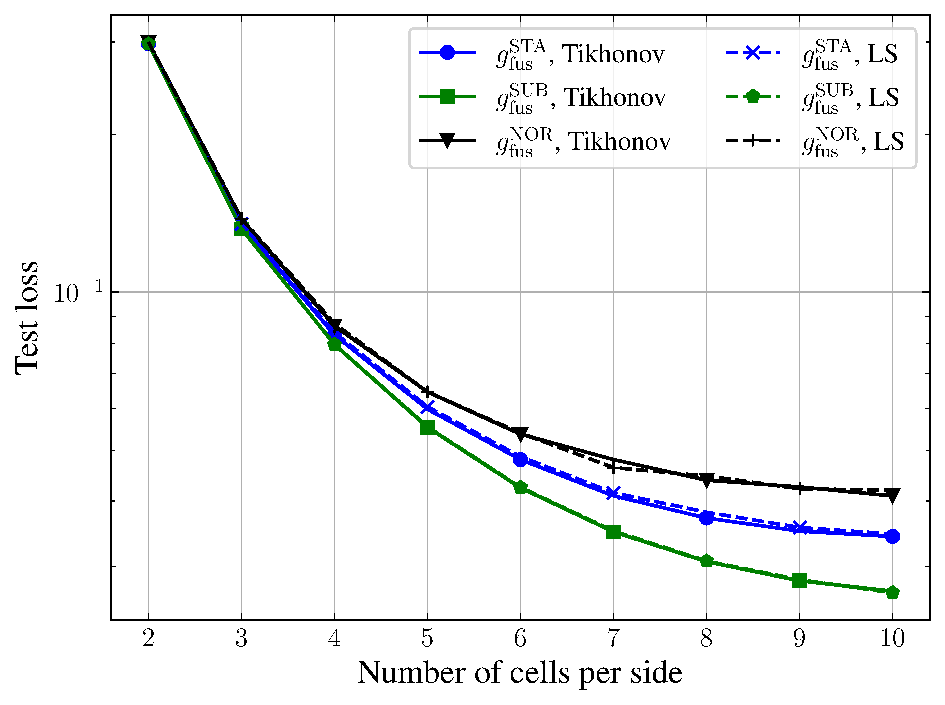}
		\caption{Performance of soft map estimations with respect to the number of cells per side.}
		\label{fig:softmap_fusion_compare}
	\end{figure}

	Figure~\ref{fig:softmap_fusion_compare} shows that as the number of cells per side increases the \ac{BCE} loss over the test dataset decreases. This is because, in the considered scene, a higher number of cells results in a greater proportion of cells that are not occupied by the target, leading to an easier overall prediction. This indicates that the dominant differences between the label and prediction in soft map representations arise from the cells near the non-zero cells of the label, which can also be observed in Fig.~\ref{fig:examples:soft_map:label}~and~\ref{fig:examples:soft_map:soft_predict}. 
	Fig.~\ref{fig:softmap_fusion_compare} also shows that, for the same map resolution, the test loss is lowest with $g_{\rm fus}^{\rm SUB}$, followed by $g_{\rm fus}^{\rm STA}$, and highest with $g_{\rm fus}^{\rm NOR}$. This is consistent with the observations in probability map and hard map estimations, as shown in Fig.~\ref{fig:probmap_fusion_compare}~and~\ref{fig:hardmap_fusion_compare}, for similar reasons. In the soft map representation, better performance means a lower test loss at the same map resolution.
	
	In addition, Figures~\ref{fig:probmap_fusion_compare},~\ref{fig:hardmap_fusion_compare},~and~\ref{fig:softmap_fusion_compare} show that the system performance with the Tikhonov-regularized channel estimator is similar to that with the \ac{LS} channel estimator when either $g_{\rm fus}^{\rm SUB}$ or $g_{\rm fus}^{\rm NOR}$ is adopted. This indicates that having either no prior knowledge or sufficient prior knowledge for environmental component removal can eliminate the performance difference between the two channel estimators. 
	Fig.~\ref{fig:probmap_fusion_compare},~\ref{fig:hardmap_fusion_compare},~and~\ref{fig:softmap_fusion_compare} also show that, when the fusion $g_{\rm fus}^{\rm STA}$ is adopted, the system performance with the Tikhonov-regularized channel estimator is strictly better than that with the \ac{LS} channel estimator. This indicates that prior knowledge of environmental components alone is insufficient to fully compensate for the information difference between the Tikhonov-regularized and \ac{LS} channel estimators.
	Therefore, in the following simulations, we use the Tikhonov-regularized channel estimator as defined in \eqref{eq:ls_est:tikhonov}.

	\subsection{Impact of the Tikhonov Regularization Factor on Sensing}\label{ssec:regu_factor_impact}
	
	\begin{figure}[!t]
		\centering
		\includegraphics[width=0.9\columnwidth]{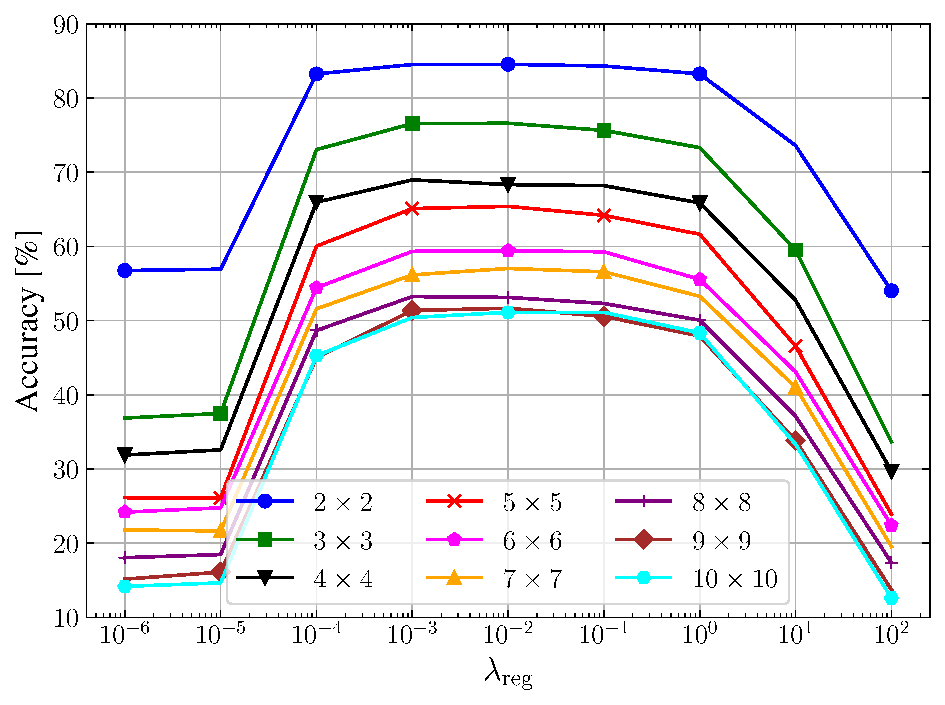}
		\caption{Impact of Tikhonov regularization factor $\lambda _{\rm reg}$ on the accuracy for probability map estimation with subtraction fusion.}
		\label{fig:regu_factor_impact}
	\end{figure}

	Different from the neural sensing estimator with \ac{LS} channel estimation, the performance of the neural sensing estimator with Tikhonov regularization is affected by the regularization factor $\lambda _{\rm reg}$. 
	Due to the similarity among the impacts of $\lambda _{\rm reg}$ on neural sensing estimators, we only present the accuracy with respect to the regularization factor $\lambda _{\rm reg}$ under different resolutions for probability map estimation with beamspace and delay-domain features and the subtraction fusion.
	
	Fig.~\ref{fig:regu_factor_impact} shows that as $\lambda _{\rm reg}$ increases, the accuracy of the probability map estimation first increases and then decreases. This is because, according to Remark~\ref{rmk:prop:tikhonov}, if $\lambda _{\rm reg}$ is either too small or too large, the Frobenius-norm constraint becomes inappropriate, leading to a suboptimal solution. This indicates the importance of selecting an appropriate value for $\lambda _{\rm reg}$ to optimize system performance.
	Fig.~\ref{fig:regu_factor_impact} also shows that the curves are similar across different map resolutions and remain approximately flat within the interval $[10^{-3}, 10^{-1}]$. This indicates that within a suitable range, small variations in $\lambda_{\rm reg}$ have a negligible impact on system performance across various map resolutions.
	
	\begin{figure}[!t]
		\centering
		\includegraphics[width=0.9\columnwidth]{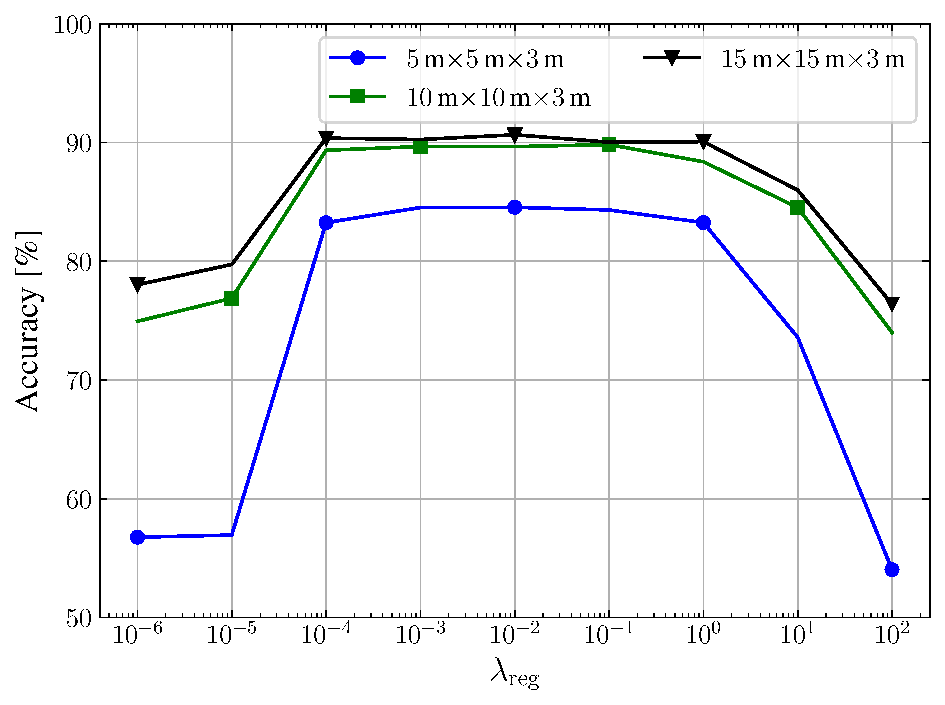}
		\caption{Impact of Tikhonov regularization factor $\lambda _{\rm reg}$ and environmental conditions on the accuracy for probability map estimation with subtraction fusion.}
		\label{fig:regu_factor_impact_diff_scene}
	\end{figure}

	To further verify that the system performance is insensitive to the variations in $\lambda _{\rm reg}$ within a suitable range, we consider different environmental conditions. 
	Specifically, with all other settings unchanged, we examined the $2\times 2$ probability map estimation with two concrete box sizes, i.e., $10\,\text{m}\times10\,\text{m}\times3\,\text{m}$ and $15\,\text{m}\times15\,\text{m}\times3\,\text{m}$.
	Fig.~\ref{fig:regu_factor_impact_diff_scene} shows that the curves are similar across the scenarios with different concrete boxes and remain approximately flat within the interval $[10^{-3}, 10^{-1}]$. This indicates that the impact of environmental conditions on the insensitivity of system performance to $\lambda _{\rm reg}$ is insignificant.
	Fig.~\ref{fig:regu_factor_impact_diff_scene} also shows that as the length and width of the concrete box increase, the accuracy of the probability map estimation improves.
	This is because, when the environment is farther from the \ac{ISAC} transceiver, the path loss along target-irrelevant reflection paths increases, and thus the target-irrelevant signal components are reduced in the received sensing signals.
	Therefore, different environmental conditions may affect the system performance, but their impact on the insensitivity of system performance to $\lambda_{\rm reg}$ is insignificant in the considered cases.

	\subsection{Impact of the Target Speed on Sensing}
	
	In this simulation, we evaluate the impact of the target speed on the system performance.
	In this case, both the sensing and communication channels are time-varying. 
	Due to the Doppler effect caused by the target velocity, the channel tensors vary across different \ac{OFDM} \acp{RB}.
	The neural sensing estimator is trained using a dataset generated from static scenes. The batch size is set to $512$.
	The test datasets are generated from dynamic scenes where the target speeds are fixed, and the velocity directions are uniformly distributed in all directions.
	Similar to Section~\ref{ssec:regu_factor_impact}, we only present the accuracy with respect to the regularization factor $\lambda _{\rm reg}$ under different resolutions for probability map estimation with beamspace and delay-domain features and the subtraction fusion.
	
	\begin{figure}[!t]
		\centering
		\includegraphics[width=0.9\columnwidth]{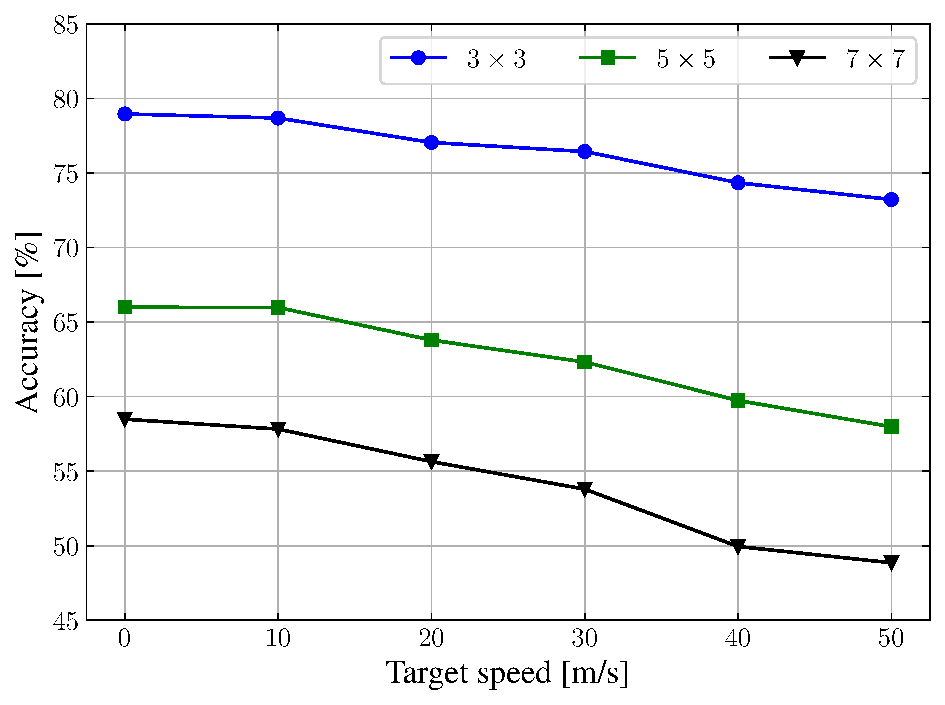}
		\caption{Impact of target speed on the accuracy for probability map estimation with subtraction fusion.}
		\label{fig:target_speed_impact}
	\end{figure}
	
	Fig.~\ref{fig:target_speed_impact} shows that the accuracy decreases as the target speed increases. This is because the Doppler effect introduces variations in the received sensing signals, which are not realized by the sensing channel estimator for time-invariant channels.  
	Fig.~\ref{fig:target_speed_impact} also shows that for target speeds below $10$\,m/s, 
	the accuracies remain nearly unchanged, and for speeds below $30$\,m/s, the accuracy drops are within $5\%$.   
	Therefore, in indoor and urban outdoor scenarios, the neural sensing estimator trained on static scenes maintains stable performance despite target motion.

	\subsection{Sensing Performance of the Proposed Neural ISAC Framework Under Different Bandwidths}
	
	\begin{figure}[!t]
		\centering
		\includegraphics[width=0.9\columnwidth]{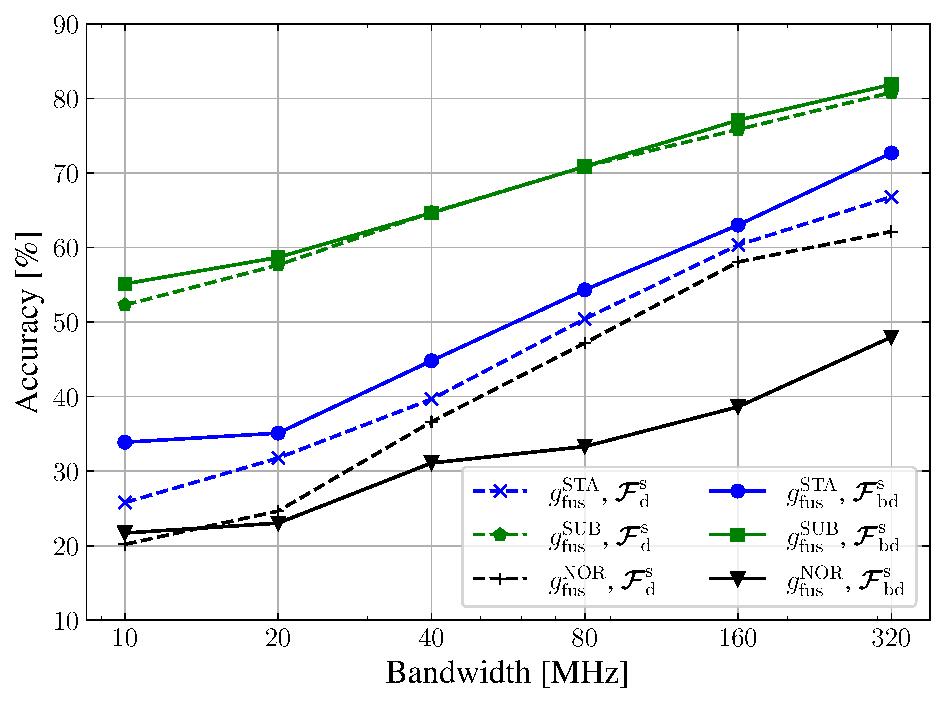}
		\caption{Performance of probability map estimation with respect to bandwidth. }
		\label{fig:bandwidth_impact:prob_map}
	\end{figure}
	
	\begin{figure}[!t]
		\centering
		\includegraphics[width=0.9\columnwidth]{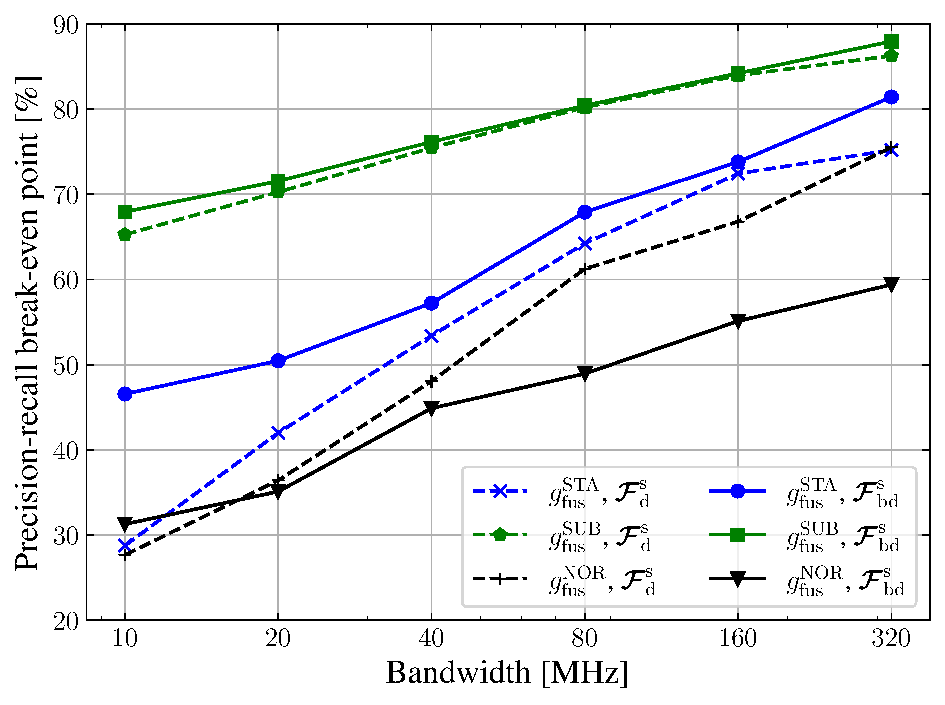}
		\caption{Performance of hard map estimation with respect to bandwidth. }
		\label{fig:bandwidth_impact:hard_map}
	\end{figure}
	
	\begin{figure}[!t]
		\centering
		\includegraphics[width=0.9\columnwidth]{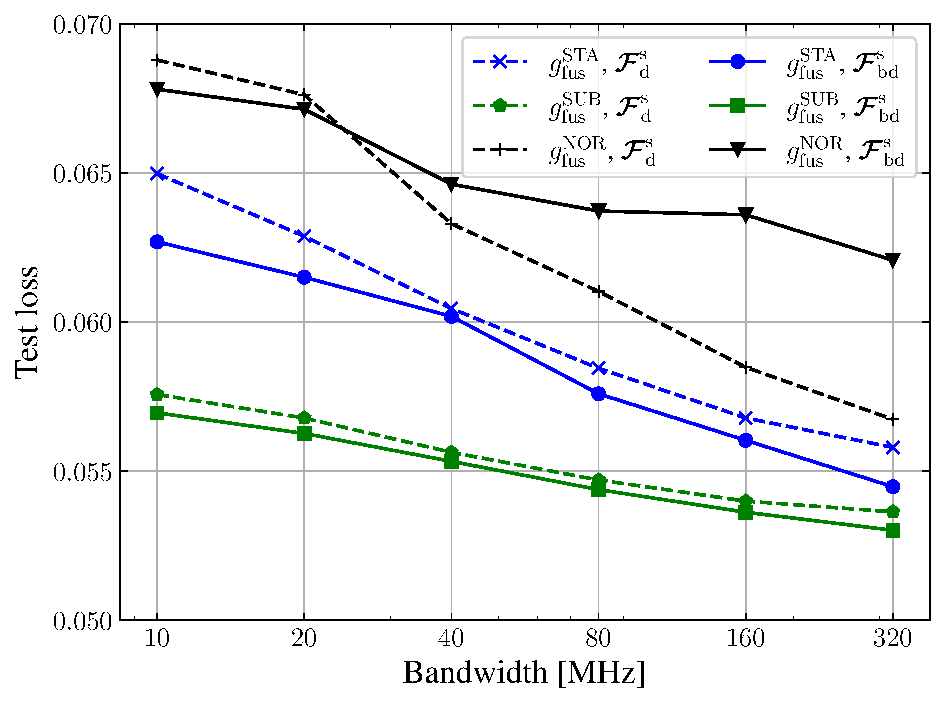}
		\caption{Performance of soft map estimation with respect to bandwidth. }
		\label{fig:bandwidth_impact:soft_map}
	\end{figure}

	In this simulation, we evaluate the system performance under different bandwidths. To ensure fairness in comparison, the number of subcarrier is fixed at $128$. Consequently, the subcarrier spacings corresponding to bandwidths of $10$\,MHz, $20$\,MHz, $40$\,MHz, $80$\,MHz, $160$\,MHz, and $320$\,MHz are the integer multiples of $78.125$\,kHz, aligning with the latest Wi-Fi standards \cite{IEEE802.11ax-2021, IEEE802.11be-D5.0}. 
	Besides, the combinations of different features and different fusions are considered. The curves related to direct features are labeled with $\Tsr F_{\rm d}^{\rm s}$, while those related to beamspace and delay-domain features are labeled with $\Tsr F_{\rm bd}^{\rm s}$.
	
	Fig.~\ref{fig:bandwidth_impact:prob_map}~and~\ref{fig:bandwidth_impact:hard_map} show that as the bandwidth increases, both accuracy and precision-recall break-even point increase. This is because the wider bandwidth enhances the time resolution in the delay domain, making the beamspace and delay-domain features more sensitive to spatial information of the scene.
	Therefore, the sensing performance can be improved by simply reusing the communication signals with more spectrum resources.
	Fig.~\ref{fig:bandwidth_impact:prob_map}~and~\ref{fig:bandwidth_impact:hard_map} show that, in the absence of prior knowledge on reference sensing \ac{CSI}, the neural \ac{ISAC} systems using direct features achieve higher accuracy and precision-recall break-even point when the bandwidth exceeds $20$\,MHz. This is because, though the beamspace and delay-domain features effectively facilitate learning in the \ac{NN}, the spatial information related to the target is obscured by the strong environmental interference. Therefore, the performance of neural \ac{ISAC} systems using beamspace and delay-domain features becomes inferior to that of systems using direct features.
	Fig.~\ref{fig:bandwidth_impact:prob_map}~and~\ref{fig:bandwidth_impact:hard_map} also show that, with the prior knowledge on reference sensing \ac{CSI}, the neural \ac{ISAC} systems using the beamspace and delay-domain features achieve consistently higher accuracy and precision-recall break-even point compared to those using the direct features. This is because the feature fusion mitigates the environmental interference in the extracted features. In this case, the \ac{NN} can learn more efficiently with the beamspace and delay-domain features than when using direct features. 
	Therefore, for scene reconstruction under probability map or hard map representation with a resolution of $5\times 5$, after mitigating environmental interference, beamspace and delay-domain features enable the neural \ac{ISAC} systems to achieve better sensing performance. However, when environmental interference remains strong and the bandwidth exceeds $20$\,MHz, the neural \ac{ISAC} systems using direct features achieve better sensing performance.
	
	Fig.~\ref{fig:bandwidth_impact:soft_map} shows that as the bandwidth increases, the \ac{BCE} loss over the test dataset decreases. Additionally, with the prior knowledge on reference sensing \ac{CSI}, the neural ISAC systems using the beamspace and delay-domain features achieve strictly lower test loss compared to those using the direct features. The reason is similar to that observed in probability map and hard map estimations. 
	Fig.~\ref{fig:bandwidth_impact:soft_map} also shows that in the absence of prior knowledge on reference sensing \ac{CSI}, the neural \ac{ISAC} systems using direct features achieve lower test loss when the bandwidth exceeds $40$\,MHz. 
	Therefore, for scene reconstruction under soft map representations with a resolution of $5\times 5$, after mitigating environmental interference, beamspace and delay-domain features enable the neural ISAC systems to achieve better sensing performance. However, when environmental interference remains strong and the bandwidth exceeds $40$\,MHz, the neural \ac{ISAC} systems using direct features achieve better sensing performance.
	
	\subsection{Limitations} 
	\label{ssec:limit}
	
	With the proposed framework, neural \ac{ISAC} systems can be trained offline using available digital twins, such as OpenStreetMap. However, the primary limitation of the proposed framework is the requirement for additional calibration phases in practical deployments, which arises from two additional considerations. First, online fine-tuning is necessary during deployment to compensate for hardware impairments. Second, acquiring prior knowledge of environmental information is important for ensuring reliable sensing performance during both the offline training and online fine-tuning phases. For instance, as shown in Fig.~\ref{fig:bandwidth_impact:prob_map}, the accuracy can degrade by more than $30$\% in the absence of prior knowledge.

	\section{Conclusions} \label{sec:conclusion}
	This paper has introduced the concept of discrete map representations for scene reconstructions and proposed a standard-compatible signal processing framework for neural \ac{ISAC} systems. 
	In particular, the proposed framework has reused \ac{MIMO}-\ac{OFDM} \ac{DL} signals to sense the scene without modifying the communication links.
	Based on the proposed discrete map representations, we have first decomposed and converted the neural \ac{ISAC} problems into multiclass or multilabel classification problems.
	Then, we have designed effective beamspace and delay-domain features, as well as feature fusion techniques that efficiently incorporate prior knowledge on environmental information, for NN training in neural \ac{ISAC} systems.
	We have shown that, in the presence of prior knowledge on reference sensing \ac{CSI}, subtraction fusion effectively removes environmental interference, including self-interference, and outperforms directly stacking features extracted from the reference sensing \ac{CSI} with those extracted from the estimated sensing \ac{CSI} as additional \ac{NN} inputs.
	For the three considered feature fusion techniques, imposing Tikhonov regularization in channel estimation ensures performance that is at least as good as using \ac{LS} estimation.
	Moreover, a larger communication bandwidth can significantly improve the sensing performance of the neural \ac{ISAC} systems.
	Furthermore, if prior knowledge on reference sensing \ac{CSI} is available, the environmental interference in the extracted features can be mitigated using the proposed feature fusion. In this case, beamspace and delay-domain features outperform direct features. Nevertheless, in the absence of prior knowledge on reference sensing \ac{CSI}, direct features provide better performance than beamspace and delay-domain features at bandwidths of $40$\,MHz or higher.
	The findings in this paper can serve as a guideline for the implementation and practical deployment of such minimally invasive neural \ac{ISAC} systems.

	\appendices
	
	\section{Proof of Proposition~\ref{prop:biased_lse}}\label{apd:proof:prop:biased_lse}
	
	The pseudo-inverse  $\left(\M P_w \M S_w\right)^\dagger$ is given by
	\begin{IEEEeqnarray}{rCl}
		\bar{\M V}_w \bar{\M \Sigma}_w^{-1} \bar{\M U}_w^\nH = \sum_{k=1}^{r_w} \left[\M \Sigma_w\right]_{k,k}\left[\M V_w\right]_{:,k} \left[\M U_w\right]_{:,k}^\nH
	\end{IEEEeqnarray}
	where the matrix $\bar{\M V}_w \in \mathbb{C} ^{L \times r_w}$ is given by $\left[\M V_w\right]_{:,1:r_w}$, the matrix $\bar{\M \Sigma}_w \in \mathbb{R} ^{r_w \times r_w}$ is given by $\left[\M \Sigma_w\right]_{1:r_w, 1:r_w}$, and the matrix $\bar{\M U}_w \in \mathbb{C} ^{N_{\rm t}\times r_w}$ is given by $\left[\M U_w\right]_{:, 1:r_w}$. In this case, the product of the matrix $\M P_w \M S_w$ and its pseudo-inverse $\left(\M P_w \M S_w\right)^{\dagger}$ is no longer a identity matrix, and the product matrix is given by
	\begin{IEEEeqnarray}{rCl}
		\M P_w \M S_w\left(\M P_w \M S_w\right)^\dagger 
		&=& \M U_w \M \Sigma_w \M V_w^\nH \bar{\M V}_w \bar{\M \Sigma}_w^{-1} \bar{\M U}_w^\nH \nonumber \\
		&=&\bar{\M U}_w \bar{\M \Sigma}_w \bar{\M V}_w^\nH \bar{\M V}_w \bar{\M \Sigma}_w^{-1} \bar{\M U}_w^\nH \nonumber\\
		&=& \bar{\M U}_w \bar{\M U}_w^\nH = \sum_{k=1}^{r_w}\left[\M U_w\right]_{:, k}\left[\M U_w\right]_{:, k}^\nH\,,
	\end{IEEEeqnarray}
	which is a projection matrix.
	Moreover, note that the \ac{LS} solution of the problem $\mathscr{P}_w^{\rm sce}$ is given by
	\begin{IEEEeqnarray}{rCl}
		\hat{\M H}_w^{\rm s} = \M Y_w^{\rm s} \left(\M P_w \M S_w\right)^\dagger\,. \label{eq:ls_est}
	\end{IEEEeqnarray}
	Substituting the signal model \eqref{eq:sys_model:sens:w} into \eqref{eq:ls_est} concludes the proof.
	
	\section{Proof of Proposition~\ref{prop:tikhonov}}\label{apd:proof:prop:tikhonov}
	For simplicity, define the objective function of $\mathscr{P}_w^{{\rm sce}(1)}$ as
	\begin{IEEEeqnarray}{rCl}
		L(\mathring{\M H}_w^{\rm s}) \triangleq\left\Vert \M Y_w^{\rm s} - \mathring{\M H}_w^{\rm s} \M P_w \M S_w \right\Vert_{\rm F}^2 + \lambda _{\rm reg} \left\Vert \mathring{\M H}_w^{\rm s} \right\Vert_{\rm F}^2\,.
	\end{IEEEeqnarray}
	According to \cite{PetPed:08}, the derivative of the objective function with respect to the parameters is given by
	\begin{IEEEeqnarray}{rCl}
		\frac{\partial L(\mathring{\M H}_w^{\rm s})}{\partial \big(\mathring{\M H}_w^{\rm s}\big)^\nH} = - 2 \M P_w \M S_w \big(\M Y_w^{\rm s} - \mathring{\M H}_w^{\rm s} \M P_w \M S_w\big)^\nH + 2 \lambda _{\rm reg} \big(\mathring{\M H}^{\rm s}_w\big)^\nH.\nonumber\\
	\end{IEEEeqnarray}
	Let ${\partial L(\mathring{\M H}_w^{\rm s})}/{\partial \big(\mathring{\M H}_w^{\rm s}\big)^\nH}\big|_{\mathring{\M H}_w^{\rm s} = \hat{\M H}_w^{\rm s}} = 0 $, then
	\begin{IEEEeqnarray}{rCl}
		\M P_w \M S_w \big(\M Y_w^{\rm s}\big)^\nH = \M P_w \M S_w \M S_w^\nH \M P_w \big(\hat{\M H}_w^{\rm s}\big)^\nH + \lambda _{\rm reg}\big( \hat{\M H}_w^{\rm s}\big)^\nH\,.\nonumber\\ \label{eq:prop:tikhonov:proof_1}
	\end{IEEEeqnarray}
	Equation~\eqref{eq:prop:tikhonov:proof_1} can be transformed to
	\begin{IEEEeqnarray}{rCl}
		\hat{\M H}_w^{\rm s}\left(\M P_w \M S_w \M S_w^\nH \M P_w^\nH  + \lambda _{\rm reg} \M I_{N_{\rm t}}\right) = \M Y_w^{\rm s} \M S_w^\nH \M P_w^\nH\,.\label{eq:prop:tikhonov:proof_2}
	\end{IEEEeqnarray}
	Multiplying both sides of equation~\eqref{eq:prop:tikhonov:proof_2} by 
	$\left(\M P_w \M S_w \M S_w^\nH \M P_w^\nH  + \lambda _{\rm reg} \M I_{N_{\rm t}}\right)^{-1}$ concludes the proof.
	
	\section{Proof of Lemma~\ref{prop:noise_avg}}\label{apd:proof:prop:noise_avg}
	Note that the entries of $\M S_w$ are drawn independently from a uniform distribution over the constellation map $\Set C$, the entries of $\M N_w^{\rm s}$ follow independently from a zero-mean complex Gaussian distribution with a variance of $N_0$, and $\M S_w$ and $ \M N_w$ are statistically independent. 
	Let
	\begin{IEEEeqnarray*}{rCl}
		\mathring{\M N}_w \triangleq \M N_{\rm w}^{\rm s} \M S_w^\nH/L\,,
	\end{IEEEeqnarray*}
	then for the $(n_{\rm r}, k)$th entry of $\tilde{\M N}_w$, its expectation is given by
	\begin{IEEEeqnarray}{rCl}
		\E{\left[\mathring{\M N}_w^{\rm s}\right]_{n_{\rm r}, k} } 
		&=& \frac{1}{L} \sum_{l=1}^{L} \E{\left[\M N_w^{\rm s}\right]_{n_{\rm r}, l}} {\left[\M S_w\right]_{k, l}^*}
		\\
		&=& 0
	\end{IEEEeqnarray}
	and its variance is given by
	\begin{IEEEeqnarray}{rCl}
		\E{\Big|\left[\mathring{\M N}_w^{\rm s}\right]_{n_{\rm r}, k}\Big|^2}
		&=&
		\frac{1}{L^2} \sum_{l=1}^{L} \E{\big|\left[\M N_w^{\rm s}\right]_{n_{\rm r}, l}\big|^2} {\big|\left[\M S_w\right]_{k, l}^*\big|^2}
		\IEEEeqnarraynumspace
		\\
		&=& \frac{ N_0 }{L}\times \frac{1}{L} \sum_{l=1}^{L}{\big|\left[\M S_w\right]_{k, l}\big|^2}
		\overset{\text{(a)}}{=} \frac{N_0}{L}  
	\end{IEEEeqnarray}
	where (a) follows from the fact that $L \gg K$ and the average energy  of symbols in $\Set C$ is normalized to $1$. Moreover, since the entries of $\M N_w^{\rm s}$ are Gaussian, the entries of $ \M N_w^{\rm s} \M S_w^\nH/L$ are also Gaussian. Hence, the entries of $ \M N_w^{\rm s} \M S_w^\nH/L$  follow independently from a zero-mean complex Gaussian with a variance of $N_0/L$.
	\bibliographystyle{IEEEtran}

	\balance
	
\end{document}